\documentclass[a4paper]{amsart}

\usepackage{mathrsfs}
\usepackage{mathtools}
\usepackage{hyperref}
\usepackage{algorithm}
\usepackage{algpseudocode}
\usepackage{linegoal}
\usepackage{dsfont}

\newcommand\mpow[2][m]{\Gamma_{#1}{#2}}
\newcommand\mlow[2][m]{\Lambda_{#1}{#2}}

\usepackage{todonotes}
\usepackage{cleveref}
\usepackage{tikz}
\usepackage{tikz-cd}

\newcommand\emptystring{}
\newcommand\rel[1]{\ensuremath{\mathbb{#1}}}

\newcommand\cate[1]{\textbf{\textup{#1}}}
\newcommand\struct[1][\emptystring]{\cate{Rel}\ifthenelse{\equal{}{#1}}{}{[#1]}}
\newcommand\Systems[1]{\cate{Strat}_{#1}}
\newcommand\MSystems[1]{\minion\operatorname{-}\Systems{#1}}
\newcommand\minion{\ensuremath{\mathscr M}}

\DeclareMathOperator{\aut}{Aut}
\DeclareMathOperator{\End}{End}
\DeclareMathOperator{\pol}{Pol}

\DeclareMathOperator{\CSP}{CSP}

\DeclareMathOperator{\PCSP}{PCSP}

\DeclareMathOperator{\BLP}{BLP}
\DeclareMathOperator{\AIP}{AIP}
\DeclareMathOperator{\sAIP}{sAIP}
\DeclareMathOperator{\dom}{dom}
\DeclareMathOperator{\im}{im}
\DeclareMathOperator{\WNU}{WNU}
\newcommand\pr[1]{\ensuremath{\textnormal{pr}_{#1}}}
\newcommand\trivgrp{*}
\newcommand\group{\mathscr G}

\newcommand\ignore[1]{}
 
\usepackage[left=3cm,right=3cm]{geometry}
\usepackage{amsthm,amssymb,cleveref}
\usepackage{thm-restate}
\declaretheorem{theorem}
\declaretheorem[numberlike=theorem]{lemma}
\declaretheorem{claim}

\declaretheorem[numberlike=theorem]{corollary}
\declaretheorem[numberlike=theorem]{proposition}
\newenvironment{claimproof}{\begin{proof}}{\end{proof}}
\theoremstyle{definition}
\declaretheorem[numberlike=theorem]{definition}

\author{Antoine Mottet}
\address{Hamburg University of Technology, Research Group for Theoretical Computer Science, Germany}
\email{antoine.mottet@tuhh.de}
\date{} \begin{document}
\title{Algebraic and algorithmic synergies between promise and infinite-domain CSPs}
\maketitle

\begin{abstract}
We establish a framework that allows us to transfer results between some constraint satisfaction problems with infinite templates and promise constraint satisfaction problems. On the one hand, we obtain new algebraic results for infinite-domain CSPs giving new criteria for NP-hardness. On the other hand, we show the existence of promise CSPs with finite templates that reduce naturally to tractable infinite-domain CSPs in the scope of the Bodirsky-Pinsker conjecture, but that are not finitely tractable, thereby showing a non-trivial connection between those two fields of research. In an important part of our proof, we also obtain uniform polynomial-time algorithms solving temporal constraint satisfaction problems.
\end{abstract}

\section{Introduction}

A constraint satisfaction problem (CSP) is a decision problem where the input consists of a set of variables taking values in a certain domain and constraints on these variables, and the problem is to decide the existence of an assignment that satisfies all the constraints.
Such problems naturally occur in many areas of computer science both in theory and in practice.
In this setting, two specific types of problems have attracted attention.

On the one hand, \emph{promise} CSPs (PCSPs), where one must decide between the existence of an assignment that satisfies a strong form of the constraints and the non-existence of an assignment that satisfies a weak form of the constraints.
A standard example to keep in mind here is the problem of deciding whether an input graph has chromatic number $\leq k$ or $>\ell$, for some fixed integers $2\leq k\leq \ell$. The complexity of this problem is open, also it is known conditionally on the 2-to-1 conjecture that the problem is NP-hard for $k\geq 4$~\cite{DinurGraphColoring}.
Typically, promise CSPs with \emph{finite} templates are considered.

On the other hand, CSPs with so-called \emph{$\omega$-categorical} constraint languages, which are infinite templates satisfying a certain ``finiteness'' condition.
This class properly contains all finite domain CSPs, for which a complexity dichotomy was proven recently~\cite{Bulatov, Zhuk}, and also contains a number of natural computational problems stemming from applications in artificial intelligence and knowledge reasoning.
A similar dichotomy is conjectured for a large subclass of these $\omega$-categorical constraint languages~\cite{BodPinPon}.

In both settings, it is possible to approach the complexity of the problems by studying the behaviour of certain algebraic objects called \emph{polymorphisms},
although the algebraic theories have essential differences. For promise CSPs, polymorphisms are functions $f\colon A^n\to B$ (for some possibly distinct sets $A,B$), and the complexity of a promise CSP is captured entirely by the so-called \emph{minor conditions} that these polymorphisms satisfy~\cite{BBKO}.
For $\omega$-categorical CSPs, polymorphisms are functions $f\colon A^n\to A$ that can be composed.
Their algebraic theory is richer due to this composition; minor conditions do \emph{not} capture the complexity of the CSPs, and more complex algebraic and topological notions need to be considered.
This is due to~\cite{TopoBirk}, and we refer the interested reader to~\cite{Willard,SymmetriesNotEnough} for further investigations on the exact structure of polymorphisms that can serve as a complexity invariant.

Despite those essential differences, we establish in our first result that CSPs within the scope of the aforementioned conjecture are computationally equivalent to a class of PCSPs with infinite left templates and finite right templates.

\begin{restatable}{theorem}{bpconj}\label{intro-pcsp-csp}
	Every problem $\CSP(\rel C)$ in the scope of the Bodirsky-Pinsker conjecture is polynomial-time equivalent to a problem of the form $\PCSP(\rel A,\rel B)$, where $\rel A$ is in the scope of the Bodirsky-Pinsker conjecture, and $\rel B$ is a finite homomorphic image thereof.
\end{restatable}

\subsection{Tractability by infinite CSPs}

A problem $\PCSP(\rel A,\rel B)$ is said to be \emph{finitely tractable} if there exists a finite template $\rel C$ with a polynomial-time tractable CSP and such that there exist homomorphisms $\rel A\to\rel C$ and $\rel C\to\rel B$.
In such a situation, $\PCSP(\rel A,\rel B)$ reduces trivially (by a ``do-nothing'' reduction) to $\CSP(\rel C)$, and therefore $\PCSP(\rel A,\rel B)$ is solvable in polynomial time, as the name suggests.
It is known that there exist finite PCSP templates with a polynomial-time tractable PCSP that are not finitely tractable~\cite{Barto}.
The relationship between finite PCSPs and finite CSPs is still not well understood. 
For example, although the condition of being finitely tractable is an algebraic invariant (i.e., something only depending on the polymorphisms of the template), an algebraic characterization of this property in terms of $\pol(\rel A,\rel B)$ is not known.
On the positive side, it is known that every finite PCSP solvable in first-order logic is finitely tractable~\cite{FOPCSP}.
On the negative side, since Barto's original example a number of other tractable templates that are not finitely tractable have been discovered~\cite{BartoAsimi}.

We are interested here in the related property of being ``tractable by an $\omega$-categorical template,'' in the sense that we are not looking for a finite template $\rel C$ as above,
but an $\omega$-categorical structure instead.
We obtain the existence of PCSP templates that are not finitely tractable, (not) solvable by particular algorithms, and are tractable by reducing to a CSP with an $\omega$-categorical template.

\begin{restatable}{theorem}{omegasandwich}\label{bw-not-finitely-tractable}
	There exist finite PCSP templates $(\rel A,\rel B)$ such that:
	\begin{itemize}
		\item $\PCSP(\rel A,\rel B)$ has width $4$,
		\item $(\rel A,\rel B)$ admits an $\omega$-categorical tractable sandwich,
		\item $(\rel A,\rel B)$ is not finitely tractable,
		\item $\PCSP(\rel A,\rel B)$ is not solvable by BLP+AIP.
	\end{itemize}
\end{restatable}

There are two parts in the statement above: the first item is about an algorithmic upper bound, while the last two items are about complexity lower bounds.
We describe in the rest of this introduction how those two parts are obtained.

\subsection{Uniform Algorithms for Temporal CSPs}

The PCSP templates we construct in~\Cref{bw-not-finitely-tractable} are derived from \emph{temporal CSPs},
which are problems where the domain is $\mathbb Q$ (or any infinite linearly-ordered set),
and where all the constraints are expressible without quantifiers using the order relation and the equality relation on $\mathbb Q$.
A standard example of a temporal CSP is the betweenness problem, where an instance is given by ternary constraints $(x,y,z)$ where the allowed orderings are $x<y<z$ and $z<y<x$.
This framework is quite large and captures a number of practical problems from the literature in artificial intelligence.
In particular, reasoning problems over Allen's interval algebra or its higher-dimensional variants as studied in~\cite{DBLP:conf/aaai/MukerjeeJ90,TractabilityResultsBlockAlgebra,ReasoningAboutTemporalRelations} can be expressed by constraints in $(\mathbb Q;<)$ using so-called first-order interpretations, and such a connection can be used to obtain complexity classifications for such problems (see e.g.~\cite{ClassificationTransferOriginal,ClassificationTransfer} and the survey by Bodirsky and Jonsson~\cite{BodirskyJonsson}).
A P/NP-complete complexity dichotomy has been shown by Bodirsky and K\'ara~\cite{BodirskyKara}, which has later been refined using the  classical dividing lines of descriptive complexity involving fixed-point logics~\cite{TemporalDescriptive}.
Another proof of the complexity dichotomy for temporal CSPs was given in~\cite{SmoothApproximations}, although relying on the same algorithms as in~\cite{BodirskyKara}.

In order to obtain the first item in~\Cref{bw-not-finitely-tractable}, we are led to investigate once more the complexity of such CSPs
and resolve some of the shortcomings of the original proof, which we describe now.
The polynomial-time algorithms given by Bodirsky and K\'ara are based on a procedure consisting in simplifying iteratively the input by removing variables  while preserving the existence or non-existence of a solution, until a point where it becomes clear that the instance is satisfiable or unsatisfiable.
The set of variables removed at each iteration can be itself obtained by solving a Boolean CSP~\cite{SmoothApproximations}.

There are two aspects of these algorithms that are unsatisfactory.
First, they only work under the assumption that the problem is tractable, as they rely on some assumptions about the polymorphisms of the template under consideration.
For the same reason, despite the recent resolution of the Feder-Vardi conjecture about finite-domain CSPs~\cite{Bulatov,Zhuk,ZhukJournal}, there is still on-going effort into trying to find a ``universal'' polynomial-time algorithm that would be complete for \emph{all} finite-domain CSPs, and where only the soundness of the algorithm uses algebraic assumptions about the template.

Second, the Bodirsky-K\'ara algorithms are difficult to use in a proof setting as they do not give a structural insight into the instances that are accepted.
In particular, this raises an issue when one is trying to expand on the Bodirsky-K\'ara complexity classification by studying CSPs where not only the order between variables is constrained, but also additional constraints are imposed, or when only particular types of instances are considered as in~\cite{ClassificationTransfer}.
More concretely, supposing that one wants to prove that a given instance $\rel X$ of a temporal CSP has a solution, it would suffice to prove that the instance is accepted by one of the algorithms solving temporal CSPs and therefore one wants to have a certificate guaranteeing such an acceptance.
However, the only certificate that an instance is accepted by the current algorithms is essentially a solution to the instance itself.

We revisit this problem by giving new algorithms solving such temporal constraints.
The algorithms can be executed on \emph{any} instance and always run in polynomial time, regardless of the complexity of the CSP.
Moreover, they are always complete, in that any instance rejected by the algorithm is certified to be unsatisfiable.
Finally, whenever the CSP is not NP-hard, then the presented algorithms are also sound and only accept instances that are satisfiable.
\begin{restatable}{theorem}{algotemporal}\label{thm:algotemporal}
	Let $\rel B$ be a temporal structure and let $\rel A$ be a finite structure that admits a homomorphism to $\rel B$.
	One of the following holds:
	\begin{itemize}
		\item $\PCSP(\rel A,\rel B)$ is solvable by local consistency or singleton AIP,
		\item or $\CSP(\rel B)$ is NP-hard.
	\end{itemize}
\end{restatable}
\Cref{thm:algotemporal} indeed provides a solution to the problems mentioned above: any instance of $\CSP(\rel B)$ with $n$ variables can be equivalently seen as an instance of $\PCSP(\rel A,\rel B)$, where $\rel A$ is a substructure of $\rel B$ of size $n$, and both local consistency and singleton AIP run in time that is a polynomial in both the instance size and the template size.

Obtaining a complete classification of the complexity of problems of the form $\PCSP(\rel A,\rel B)$ where $\rel B$ is a temporal structure would be an interesting result that we leave for future research.

\subsection{Identities}

Finally, we obtain some results allowing to relate the algebraic properties of $\rel C$ and those of certain PCSP templates derived from $\rel C$.
Those identity transfers are used to prove that the templates exhibited in~\Cref{bw-not-finitely-tractable} are not solvable by some of the standard algorithmic methods in the area of promise constraint satisfaction and are not finitely tractable, see~\Cref{sect:hardness}.
Moreover, using those transfers we obtain some new unconditional hardness criteria for CSPs with an $\omega$-categorical template.
\begin{theorem}\label{thm:olsak-wnu}
	Let $\rel C$ be a structure and let $\group\leq\aut(\rel C)$ be an oligomorphic subgroup. The following hold:
	\begin{itemize}
	\item	If $\pol(\rel C)$ does not satisfy the Ol\v{s}\'ak identities modulo $\overline{\group}$, then $\CSP(\rel C)$ is NP-hard.
	\item If $\pol(\rel C)$ does not satisfy the WNU identities for some arity $m\geq 3$ modulo $\overline{\group}$, then $\CSP(\rel C)$ is not solvable by local consistency methods.
	\end{itemize}
\end{theorem}

Conditionally on the hardness of $3$- vs. $\ell$-coloring, we obtain the following.
\begin{theorem}\label{no-6-siggers-hard}
	Suppose that $\PCSP(\mathbb K_3,\mathbb K_\ell)$ is NP-hard for all $\ell\geq 3$.
	Let $\rel C$ be a structure and let $\group\leq\aut(\rel C)$ be an oligomorphic subgroup.
	If\/ $\pol(\rel C)$ does not satisfy the 6-ary Siggers identity modulo $\overline{\group}$, then $\CSP(\rel C)$ is NP-hard.
\end{theorem}

Due to space restrictions, some proofs are omitted. All results mentioned here are original unless stated otherwise.
 
\section{Definitions}\label{sect:defs}

For $n\geq 1$ we denote by $[n]$ the set $\{1,\dots,n\}$.
For $i_1,\dots,i_k\in\{1,\dots,n\}$ and an $n$-tuple $a$, we write $\pr{i_1,\dots,i_k}(a)$ for the tuple $(a_{i_1},\dots,a_{i_k})$.

All relational structures considered here are at most countable and have finite signature.
Blackboard bold letters are used to denote structures, while the corresponding standard font letters are used to denote their domains.
A \emph{homomorphism} from a structure $\rel A$ to a structure $\rel B$ is a map $h\colon A\to B$ such that for every relation $R$ in the signature of $\rel A$ and $\rel B$, and all $(a_1,\dots,a_r)\in R^{\rel A}$, one has $(h(a_1),\dots,h(a_r))\in R^{\rel B}$.
We write $\rel A\to\rel B$ if there exists a homomorphism from $\rel A$ to $\rel B$.
A \emph{partial} homomorphism is a homomorphism from an induced substructure of $\rel A$ to $\rel B$.

For a fixed pair $(\rel A,\rel B)$ of structures such that $\rel A$ admits a homomorphism to $\rel B$, the problem $\PCSP(\rel A,\rel B)$ is the promise problem of deciding whether an input structure $\rel X$ admits a homomorphism to $\rel A$ or no homomorphism to $\rel B$.
The pair $(\rel A,\rel B)$ is called the \emph{template} of the problem.
When $\rel A=\rel B$, we denote this problem by $\CSP(\rel A)$.

A \emph{primitive positive (pp) formula} is a first-order formula that is built by using atomic formulas, conjunctions, and existential quantifications only.
Given a $d\geq 1$ and $R\subseteq (A^d)^r$, we say that $R$ is pp-definable in $\rel A$ if there exists a pp-formula $\varphi$ with $d\times r$ free variables such that $(a^1,\dots,a^r)\in R$ holds if, and only if, $\rel A\models \varphi(a^1_1,\dots,a^1_d,\dots,a^r_1,\dots,a^r_d)$.
We say that $(\rel C,\rel D)$ is a \emph{pp-power} of $(\rel A,\rel B)$ if $\rel C=(A^d;R_1,\dots,R_k), \rel D=(B^d;S_1,\dots,S_k)$ for some $d\geq 1$ and for all $i$ there exists a pp-formula $\varphi_i$ such that $R_i$ is the relation defined by $\varphi_i$ in $\rel A$ and $S_i$ is the relation defined by $\varphi_i$ in $\rel B$.
We say that $(\rel C,\rel D)$ is a \emph{homomorphic relaxation} of $(\rel A,\rel B)$ if there exist homomorphisms $\rel C\to\rel A$ and $\rel B\to\rel D$.
Finally, we say that $(\rel C,\rel D)$ is \emph{pp-constructible} in $(\rel A,\rel B)$ if it is a homomorphic relaxation of a pp-power of $(\rel A,\rel B)$.
An easy observation is that if $(\rel C,\rel D)$ is pp-constructible in $(\rel A,\rel B)$ then  $\PCSP(\rel C,\rel D)$ reduces to $\PCSP(\rel A,\rel B)$ by a polynomial-time reduction.

If $\theta$ is an equivalence relation on $\rel B$, $\rel B/\theta$ denotes the structure %
 whose domain consists of the equivalence classes of $\theta$ and whose relations are the images of the relations of $\rel B$ under the canonical projection.
If $\theta$ is an equivalence relation on an induced substructure $\rel C$ of $\rel B$, we abuse notation and write $\rel B/{\theta}$ for the structure $\rel C/{\theta}$.
If $\group\subseteq\aut(\rel B)$, we denote by $\rel B/{\group}$ the quotient of $\rel B$ by the relation containing pairs $(a,b)$ where $a,b$ are in the same orbit under $\group$.

\subsection{Countably categorical structures}

If $\group$ is a group of permutations on a set $B$, we denote the set of orbits of $B^k$ under the natural action of $\group$ on $k$-tuples by $B^k/{\group}$.
We say that $\rel B$ is \emph{$\omega$-categorical} if the set $B^k/{\aut(\rel B)}$ is finite for all $k$, in which case we also say that $\aut(\rel B)$ is \emph{oligomorphic}.
A \emph{reduct} of $\rel B$ is a structure $\rel B'$ with the same domain and such that every relation of $\rel B'$ admits a first-order definition in the structure $\rel B$.
Note that this implies that $\aut(\rel B)\subseteq\aut(\rel B')$ and therefore if $\rel B$ is $\omega$-categorical, so is $\rel B'$.
We say that $\rel B$ is \emph{finitely bounded homogeneous} if: every two isomorphic finite substructures of $\rel B$ are in the same orbit under $\aut(\rel B)$, and there is a bound $k\in\mathbb N$ such that every structure that does \emph{not} embed into $\rel B$ contains a $k$-element structure that also does not.
Every finitely bounded homogeneous structure is $\omega$-categorical,
and Bodirsky and Pinsker conjectured that the CSP of a reduct of a finitely bounded homogeneous structure is in P or NP-complete.

\subsection{Polymorphisms}

A \emph{polymorphism} of the template $(\rel A,\rel B)$ is a homomorphism $\rel A^n\to\rel B$.
The set $\pol(\rel A,\rel B)$ of all polymorphisms has the following algebraic structure: for every $f\in\pol(\rel A,\rel B)$ of arity $n$, every $\sigma\colon\{1,\dots,n\}\to\{1,\dots,m\}$, and every $u\in\End(\rel B)$, the operations defined by $(a_1,\dots,a_m)\mapsto f(a_{\sigma(1)},\dots,a_{\sigma(n)})$ and $u\circ f$ are also elements of $\pol(\rel A,\rel B)$. We denote the former by $f^{\sigma}$, and this notation generalizes naturally to functions $A^X\to B$ and $\sigma\colon X\to Y$.
Let $\minion{}$ be a set of operations $A^X\to B$ for possibly different finite sets $X$.
We denote by $\minion{X}$ the set of all $f\in\minion{}$ with domain $A^X$. 
We say that $\minion{}$ is a minion if it is closed under the operations $f\mapsto f^\sigma$ for all $f\in\minion{[n]}$ and all $\sigma\colon [n]\to [m]$.

\subsection{Algorithmic notions}

In the following, $\rel A$ is a finite structure and $\rel X$ is an instance of $\CSP(\rel A)$.

\subsubsection{Local consistency}
Let $\rel X$ be an instance of $\CSP(\rel A)$, and let $k\geq 1$.
A \emph{potato system} from $\rel X$ to $\rel A$ is a collection $(D_x)_{x\in X}$ of non-empty subsets of $A$ such that for every $(x_1,\dots,x_k)\in R^{\rel X}$, and every $i\in\{1,\dots,k\}$ and $a\in D_{x_i}$, there exists $t\in R^{\rel A}$ such that $t_i=a$ and $t_j\in D_{x_j}$ for all $j\in\{1,\dots,k\}$.
We say that $\rel X$ is \emph{singleton arc consistent} with respect to $\rel A$ if there exists a potato system $(D_x)_{x\in X}$ from $\rel X$ to $\rel A$ such that for every $x\in X$ and $a\in D_x$, there exists a potato system $(D'_y)_{y\in X}$ from $X$ to $A$ such that  $D'_y\subseteq D_y$ for every $y\in X$ and $D'_x=\{a\}$.

A \emph{non-trivial k-strategy} from $\rel X$ to $\rel A$ is a non-empty collection $H$ of partial homomorphisms from $\rel X$ to $\rel A$ satisfying the following conditions:
\begin{itemize}
	\item if $f\in H$ and $K\subseteq\dom(f)$, then $f|_K \in H$,
	\item if $f\in H$, $|\dom(f)| < k$, and $a\in X$, there exists $g\in H$ such that $\dom(g)=\dom(f)\cup\{a\}$ and $g$ extends $f$.
\end{itemize}
We write $H_{x,y}$ for $\{ (f(x),f(y)) \mid f\in H, \dom(f)=\{x,y\}\}$, and similarly  $H_x=\{f(x) \mid f\in H, x\in\dom(f)\}$.

We say that $\rel X$ is $k$-consistent with respect to $\rel A$ if there exists a non-trivial $k$-strategy from $\rel X$ to $\rel A$.
The existence of a non-trivial $k$-strategy from $\rel X$ to $\rel A$ can be decided by the standard $k$-consistency algorithm in time polynomial in both $|\rel X|$ and $|\rel A|$.
The union of two strategies is a strategy, and therefore there always exists an inclusion-maximal strategy $H$ from $\rel X$ to $\rel A$.
The maximality implies that if $h\colon\rel X\to\rel A$ is a homomorphism, then $h|_K\in H$ for every $K\subseteq X$ of size at most $k$.
In particular, if there exists a homomorphism $h\colon\rel X\to\rel A$, then there exists a non-trivial strategy.

A semilattice operation is an operation $f\colon A^2\to A$ that is commutative, associative, and such that $f(a,a)=a$ for all $a\in A$.
It is known that if $\rel A$ admits a semilattice polymorphism, then every instance $\rel X$ of $\CSP(\rel A)$ with a potato system admits a homomorphism $h\colon\rel X\to\rel A$ (see, e.g.,~\cite{FederVardi}).

Suppose that $k$ is larger than the arity of every relation in the signature of $\rel X$ and $\rel A$.
If there exists a non-trivial $k$-strategy $H$ from $\rel X$ to $\rel A$, then it is folklore and easy to prove that $\rel X$ is singleton arc consistent with respect to $\rel A$, where one can take $D_x=H_x$ for all $x\in X$ and $D'_y=\{b \mid (a,b)\in H_{x,y}\}$ for all $a\in D_x$ and $y\in X$.

Finally, $(\rel A,\rel B)$ is said to have width $k$ if every instance $\rel X$ admitting a non-trivial $k$-strategy to $\rel A$ admits a homomorphism to $\rel B$; we say $(\rel A,\rel B)$ is solvable by local consistency if it has width $k$ for some $k$;
similarly, we say that $(\rel A,\rel B)$ is solvable by singleton arc consistency if every instance singleton arc consistent with respect to $\rel A$ admits a homomorphism to $\rel B$.

\subsubsection{Relaxations}
The existence of a homomorphism $\rel X\to\rel A$ can be expressed as the existence of a $\{0,1\}$ solution to the following finite system of equations.
The variables are as follows: for each $x\in \rel X$ and $a\in A$, one has a variable $\lambda_x(a)$, and for each relation symbol $R$ and $y\in R^{\rel X}$ and $b\in R^{\rel A}$ one has a variable $\lambda_y(b)$. The equations are:
\begin{align*}
	\sum_{a\in A} \lambda_x(a) &=1 & \parbox{5cm}{for all $x\in \rel X$}\\
	\sum_{b\in R^{\rel A}} \lambda_y(b) &= 1 & \parbox{5cm}{for all $y\in R^{\rel X}$} \\
	\sum_{\substack{b\in R^{\rel A} \\ b_i = a}} \lambda_y(b) &= \lambda_{x_i}(a)  & \parbox{5cm}{for all $y=(x_1,\dots,x_r)\in R^{\rel X}$,\\ $i\in\{1,\dots,r\}$, and $a\in A$}\\
\end{align*}
The affine integer programming relaxation of $\rel X$, denoted by $\AIP(\rel X,\rel A)$, is the same system where the variables are now considered as integer variables, while the linear programming relaxation $\BLP(\rel X,\rel A)$ denotes this system where the variables are considered as non-negative rationals.
It is possible to check in polynomial time in $|\rel X|$ and $|\rel A|$ whether $\BLP(\rel X,\rel A)$ and $\AIP(\rel X,\rel A)$ admit a solution~\cite{LP,HermiteSmith}.

Let $\sAIP(\rel X,\rel A)$ be a shorthand for the property that there is a collection $(D_x)_{x\in X}$ of non-empty subsets of $A$ such that for every $x\in X$ and every $a\in D_x$, there exists a solution to $\AIP(\rel X,\rel A)$ such that $\lambda_x(a)=1$, $\lambda_x(a')=0$ for all $a'\neq a$, and for all $x'\in X$ and $a'\not\in D_{x'}$, one has $\lambda_{x'}(a')=0$.
If there exists a homomorphism $h\colon\rel X\to\rel A$, then $\sAIP(\rel X,\rel A)$ holds by taking $D_x$ to be the set $\{h(x)\}$ for all $x\in X$.
We say that $\PCSP(\rel A,\rel B)$ is solved by singleton AIP if whenever $\sAIP(\rel X,\rel A)$ holds, then there exists a homomorphism $\rel X\to\rel B$. The condition $\sAIP(\rel X,\rel A)$ can be checked in polynomial time, since it can be checked by solving a polynomial amount of affine integer systems of equations.
It is known that if $\rel A$ has two elements and admits the ternary minority operation as a polymorphism, then $\CSP(\rel A)$ is solvable by AIP~\cite{BBKO}.

The algorithm called ``BLP+AIP'' works by first finding an interior point for $\BLP(\rel X,\rel A)$, setting variables to $0$ when they are not included in the support of an interior point, and then solving the modified $\AIP(\rel X,\rel A)$~\cite{BLPAIP}.
We say that $\PCSP(\rel A,\rel B)$ is solvable by BLP+AIP if every instance $\rel X$ for which solutions as above exist admits a homomorphism to $\rel B$. 
\section{Uniform Algorithms for Temporal Constraints}\label{sect:algorithms}
Let $\rel B$ be a temporal structure, i.e., a structure over the set $\mathbb Q$ and whose relations are all definable by a boolean combinations of atomic formulas of the form $x<y$.
We assume here that $\rel B$ has $<$ as part of its relations, and that if we let $m$ be the maximal arity of the relations of $\rel B$, then all relations of arity at most $m$ that are pp-definable in $\rel B$ are themselves relations of $\rel B$.
Since $\rel B$ is $\omega$-categorical, there are only finitely many such relations.

We let $\Theta$ be the equivalence relation on $\mathbb Q_{\geq 0}$ with the two equivalence classes $\{0\}$ and $\mathbb Q_{>0}$,
and we consider the expansion $\rel B^*=(\rel B,\{0\},\Theta)$ of $\rel B$ by the unary relation $\{0\}$ and the binary relation $\Theta$.
We denote by $\mathscr C$ the set of all functions $f\in\pol(\rel B^*)$ seen as functions on set of equivalence classes of $\Theta$ (i.e., $\mathscr C$ is the image of $\pol(\rel B^*)$ under the function that maps $f$ to the function that $f$ induces on the 2-element set $\mathbb Q_{\geq 0}/{\Theta}$).

It is proved in~\cite{SmoothApproximations} that if $\CSP(\rel B)$ is not NP-hard, then $\mathscr C$ is a non-trivial clone and, by the classical result of Post~\cite{post1941two}, it contains one of 3 types of operations: a binary semilattice operation, the ternary majority operation, or the ternary minority operation.
As it turns out, the majority case does not need to be handled separately in this setting:
\begin{restatable}{lemma}{majsemilattice}\label{maj-implies-semilattice}
	If\/ $\mathscr C$ contains the majority operation, then it contains a semilattice operation.
\end{restatable}

\begin{definition}[Similarly as~\cite{TemporalDescriptive}]
Let $\rel X$ be an instance of $\CSP(\rel B)$ with domain $X$ and let $F\subseteq X$.
We define the \emph{projection of $\rel X$ onto $X\setminus F$}, denoted by $\rel X\setminus F$, and the \emph{contraction} of $F$ in $\rel X$, denoted by $\rel X/{F}$, as structures in the same signature as $\rel B$ and defined as follows: %
\begin{itemize}
\item The domain of $\rel X\setminus F$ is $X\setminus F$. For every relation symbol $R$ of arity $m$, and every $(x_1,\dots,x_m)\in R^{\rel X}$, let $x_{i_1},\dots,x_{i_r}$ be an enumeration of the elements of $\{x_1,\dots,x_m\}\setminus F$. Let $\tilde R$ be the symbol corresponding to the relation defined by $\exists_{y\in F} y\, (R(x_1,\dots,x_m))$ in $\rel B$.
Such a symbol exists by our assumption that all relations that have small arity and a pp-definition in $\rel B$ are part of the relations of $\rel B$.
Then $\tilde R^{\rel X\setminus F}$ contains $(x_{i_1},\dots,x_{i_r})$.

\item The domain of $\rel X/F$ is $X$. For every relation symbol $R$ of arity $m$, and every $(x_1,\dots,x_m)\in R^{\rel X}$, let $(x_1,\dots,x_m)\in\tilde R^{\rel X/{F}}$, where $\tilde R$ is the symbol corresponding to the relation defined by $R(x_1,\dots,x_m)\land\bigwedge_{z,z'\in F}z=z'$ in $\rel B$.
\end{itemize}
\end{definition}

We say that a $k$-strategy $H$ from $\rel X$ to a temporal structure is \emph{well-ordered} if every non-empty subset of $H_x$ admits a minimal element; $\rel X$ is \emph{well-ordered} $k$-consistent if it has a non-trivial well-ordered $k$-strategy.
Similarly, we say that $\rel X$ is \emph{well-ordered} (singleton) arc consistent if the potato system witnessing the (singleton) arc consistency  consists of sets whose non-empty subsets have a minimal element.

\begin{restatable}{lemma}{projectionstrategy}\label{projection-strategy}
	Let $\rel B$ be a temporal structure with relations of arity at most $k$.
	Let $\rel X$ be an instance, let $H$ be a $k$-strategy from $\rel X$ to $\rel B$, and let $F\subseteq X$.
	The following hold:
	\begin{enumerate}
	\item $G=\{ h|_{K\cap (X\setminus F)} \mid h\in H, \dom(h)=K\}$
		is a $k$-strategy from $\rel X\setminus F$ to $\rel B$. If $H$ is well-ordered, so is $G$.
	\item If $H_{x,y}$ is a subset of the diagonal for all $x,y\in F$, then $H$ is a $k$-strategy from $\rel X/{F}$ to $\rel B$.
	\item If $\rel X$ is (well-ordered) singleton arc consistent with respect to $\rel B$, so is $\rel X\setminus F$.
	\item Let $\rel A$ be a finite substructure of $\rel B$. If $\sAIP(\rel X,\rel A)$ holds, then $\sAIP(\rel X\setminus F,\rel A)$ holds.
	\end{enumerate}
\end{restatable}

A \emph{free set} for a structure $\rel X$ is the preimage of $\{0\}$ under a homomorphism $\rel X\to \rel B/{\Theta}$. %
Whenever $\rel X$ is an instance that is singleton arc consistent, and $\mathscr C$ contains a semilattice operation, then $\rel X$ admits a non-empty free set.
\begin{lemma}\label{semilattice-sac-free-set}
	Let $\rel B$ be a temporal structure.
	Let $\rel X$ be an instance of\/ $\CSP(\rel B)$ with $X\neq\emptyset$ and suppose that $\mathscr C$ contains a semilattice operation.
	Let $(D_x)_{x\in X}$ be a well-ordered potato system witnessing singleton arc consistency of $\rel X$.
	Let $a$ be the minimal element of $\bigcup_{x\in X} D_x$.
	Then $\rel X$ admits a non-empty free set $F$ where $a\in D_x$ for every $x\in F$.
\end{lemma}
\begin{proof}
	By assumption and~\Cref{easylemma}, $\rel B/{\Theta}$ has a semilattice polymorphism.
	Consider an arbitrary $x\in X$ such that $a \in D_x$.
	Since $\rel X$ is singleton arc consistent, it has a non-trivial potato system $(D'_{y})_{y\in X}$ such that $D'_x=\{a\}$.
	Note that for every injective monotone map $\alpha\colon \mathbb Q\to \mathbb Q$, we get that $(\alpha(D'_{y}))_{y\in X}$ is a well-ordered potato system from $\rel X$ to $\rel B$.
	Let us pick such an $\alpha$ such that $\alpha(a)=0$.
	
	The system $(\alpha(D'_y))_{y\in X}$ can be composed with the projection $\rel Q_{\geq 0}\to\rel Q_{\geq 0}/{\Theta}$ to obtain a non-trivial potato system for $\rel X$ with respect to $\rel B/{\Theta}$, from which it follows that there exists a homomorphism $h\colon\rel X\to\rel B/{\Theta}$ such that $h(y)\in \alpha(D'_y)/{\Theta}$ for all $y\in X$.
	In particular, $h(x)= \{0\}$, and $h^{-1}(\{0\})$ is a non-empty free set for $\rel X$.
	Moreover, if $x\in F$ then $0\in \alpha(D'_x)$ and thus $a\in D'_x$.
\end{proof}

It is known~\cite{BodirskyKara,SmoothApp-journal} that if $\CSP(\rel B)$ is not NP-hard, then $\pol(\rel B^*)$ contains one of the two binary operations denoted by $pp$ or $ll$. The latter is an injective operation satisfying $0=ll(0,0)<ll(0,y)<ll(x,y')$ for all $x,y\geq 0$ and $y'\in\mathbb Q$.
In particular, $ll$ induces a semilattice operation in $\mathscr C$.
On $(\mathbb Q_{\geq 0})^2$, $ll(x,y)<ll(x',y')$ iff $y<y'$ or ($y=y'$ and $x<x'$).
Apart from the properties above, no additional knowledge about these functions is necessary.
\vspace{-1em}

\subsection{Templates with $pp$}

\begin{lemma}[\cite{BodirskyKara}]\label{induction-pp}
Let $\rel B$ be a temporal structure such that $pp\in\pol(\rel B)$.
Let $\rel X$ be an instance of $\CSP(\rel B)$ that has a free set $F$. If $\rel X\setminus F\to \rel B$, then $\rel X\to\rel B$.%
\end{lemma}

This is enough to obtain that, if $pp\in\pol(\rel B)$ and $\mathscr C$ has a semilattice operation, then every instance of $\CSP(\rel B)$ that is boundedly  singleton arc consistent has a solution.

\begin{theorem}\label{pp-semilattice}
	Let $\rel B$ be a temporal structure. Suppose that $pp\in\pol(\rel B)$ and that $\mathscr C$ has a semilattice operation.
	Let $\rel X$ be a well-ordered singleton arc consistent instance of $\CSP(\rel B)$.
	Then there exists a homomorphism $\rel X\to\rel B$.
\end{theorem}
\begin{proof}
	We proceed by induction on the size of $\rel X$, the case where $\rel X$ has empty domain being trivial.
	Assume that $X\neq\emptyset$.
	By~\Cref{semilattice-sac-free-set},  $\rel X$ admits a non-empty free set $F$.
	The projection of $\rel X$ onto $X\setminus F$ is then itself singleton arc consistent by~\Cref{projection-strategy}, and must therefore have a homomorphism to $\rel B$ by induction hypothesis.
	Since $pp\in\pol(\rel B)$,~\Cref{induction-pp} gives that $\rel X\to\rel B$ as well.
\end{proof}

We now turn to the case where $\pol(\rel B)$ contains $pp$ and $\mathscr C$ contains the minority operation.

\begin{theorem}\label{pp-maltsev}
	Let $\rel B$ be a temporal structure. Suppose that $pp\in\pol(\rel B)$ and that $\mathscr C$ has a minority operation.
	Let $\rel X$ be an instance of $\CSP(\rel B)$.
	If $\sAIP(\rel X,\rel A)$ holds for some finite substructure $\rel A$ of $\rel B$,  there exists a homomorphism $\rel X\to\rel B$.
\end{theorem}
\begin{proof}
	Since $\mathscr C\subseteq\pol(\rel A/{\Theta})$ by~\Cref{easylemma} and $\mathscr C$ has a minority operation, $\CSP(\rel A)$ is solvable by $\AIP$.
	Since $\sAIP(\rel X,\rel A)$ holds, there is as in~\Cref{semilattice-sac-free-set} a non-empty free set $F$ for $\rel X$.
	The projection $\rel Y$ of $\rel X$ onto $X\setminus F$ is such that $\sAIP(\rel Y,\rel A)$ holds by~\Cref{projection-strategy}.
	By induction hypothesis we get a homomorphism $\rel Y\to\rel B$.
	By~\Cref{induction-pp}, there exists a homomorphism $\rel X\to\rel B$.%
\end{proof}

\subsection{Templates with $ll$}

The strategy we employ here is similar as for the previous section, except that the induction argument is slightly more involved than an application~\Cref{induction-pp} as in the previous case.
Since $ll$ is an injective operation, we must care about the kernel of $h\colon\rel Y\to\rel B$ to lift it to a homomorphism $\rel X\to\rel B$.

\begin{definition}
Let $\rel B$ be a temporal structure.
Let $\rel X$ be an instance of $\CSP(\rel B)$ with domain $X$ and let $H$ be a bounded $k$-strategy from $\rel X$ to $\rel B$ and let $a$ be the minimal element in $\bigcup_{x\in X}H_x$.
A \emph{decomposition sequence} for $(\rel X, H)$ is defined recursively as follows:
\begin{itemize}
	\item If $X$ is itself a free set, then $X$ is a decomposition sequence (of length 1).
	\item Suppose that $F_1$ is an inclusion-minimal free set for $\rel X$ such that $a\in H_x$, for every $x\in F_1$, and such that $H_{x,y}$ is a subset of the diagonal for every $x,y\in F_1$. Let $F_2,\dots,F_k$ be a decomposition sequence of $(\rel Y, G)$, where $\rel Y=(X/F_1)\setminus F_1$ is the contracted projection of $\rel X$ onto $X\setminus F_1$ and $G$ is defined as in~\Cref{projection-strategy} (1). Then $F_1,\dots,F_k$ is a decomposition sequence of $\rel X, H$.
\end{itemize}
\end{definition}

As in the case of $pp$, we first show that if $\rel X$ satisfies a consistency condition, then there exists a decomposition sequence, and we then show that such a decomposition sequence yields a homomorphism $\rel X\to\rel B$.

\begin{lemma}\label{ll-min-freeset}
	Let $\rel B$ be a temporal structure with relations of arity at most $k\geq 3$ and suppose that $ll\in\pol(\rel B)$.
	Let $\rel X$ be an instance of\/ $\CSP(\rel B)$ with a non-trivial well-ordered $k$-strategy $H$.
	Let $F\subseteq X$ be an inclusion-minimal non-empty free set for $\rel X$.
	For every $x,y\in F$ and $(a,b)\in H_{x,y}$ we have $a=b$.
\end{lemma}
\begin{proof}
	For the following proof, given $a\in\mathbb Q$ and $H\subseteq\mathbb Q\times\mathbb Q$, we write $a+H$ for the set $\{b\in\mathbb Q\mid (a,b)\in H\}$.

	Suppose that the conclusion does not hold.
	Consider a minimal element $a$ such that there exist $x, y\in F$ and $(a,b)\in H_{x,y}$ such that $a<b$.
	Such a minimal element exists since $H$ is well-ordered.
	Let $G$ consist of the elements $u\in F$ such that $a\in H_u$. 
	For $u,v\in G$ define $u\rightsquigarrow v$ if $a+H_{u,v}=\{a\}$.
	
	\begin{claim}\label{claim:lowerbound} If $(a,b)\in H_{u,v}$ for $u,v\in F$, then $b\geq a$.
	\end{claim}
	\begin{claimproof}
	Otherwise, $(b,a)\in H_{v,u}$ with $v,u\in F$, a contradiction to the minimality of $a$.
	\end{claimproof}
	
	\begin{claim}\label{claim:transitive} $\rightsquigarrow$ is reflexive and transitive on $G$.\end{claim}
	\begin{claimproof}
	The reflexivity is clear, since for every $u\in G$ we have $a\in H_u$ and $H_{u,u}$ is equal to the diagonal on $H_u$ so that in particular $a+H_{u,u}=\{a\}$.
	Suppose now that $u,v,w\in G$ and $u\rightsquigarrow v\rightsquigarrow w$, and let $(a,c)\in H_{u,w}$.
	Then since $D$ is a $k$-strategy there exists $b$ such that $(a,b)\in H_{u,v}$ and $(b,c)\in H_{v,w}$.
	Since $u\rightsquigarrow v$, we get $b=a$, and since $v\rightsquigarrow w$, we get $c=a$.
	Thus $a+H_{u,w}\subseteq \{a\}$, and it must be non-empty since $H$ is a $k$-strategy and $a\in H_u$.
	\end{claimproof}
	
	Let $F'$ be a sink strongly connected component of the graph $(G,\rightsquigarrow)$.
	
	\begin{claim}\label{claim:proper}$F'\subsetneq F$.\end{claim}
	\begin{claimproof} If $F=F'$ (which also implies $G=F$) then $(G,\rightsquigarrow)$ is strongly connected, and since $\rightsquigarrow$ is transitive by~\Cref{claim:transitive} it must be equal to $F^2$.
	However, by assumption, there exist $x,y\in F$ and $b>a$ such that $(a,b)\in H_{x,y}$, and therefore $x\not\rightsquigarrow y$.
	\end{claimproof}
	
	\begin{claim}\label{claim:greater} For every $u\in F',v\in F\setminus F'$, there exists $b$ such that $a<b$ and $(a,b)\in D_{u,v}$.\end{claim}
	\begin{claimproof}
	Suppose first that $a\not\in D_{v}$.
	Let $b\in D_v$ be arbitrary such that $(a,b)\in D_{u,v}$.
	By~\Cref{claim:lowerbound}, we have $b\geq a$, so $b>a$ and we are done.
	Otherwise, suppose that $a\in D_{v}$ and therefore $v\in G$.
	Since $u\not\rightsquigarrow v$ (as otherwise $v$ would be in $F'$ since it is a sink component), we have $a+H_{u,v}\neq\{a\}$ and thus there exists $b\neq a$ such that $(a,b)\in H_{u,v}$. Since $b\geq a$ by~\Cref{claim:lowerbound}, we get $b>a$.\end{claimproof}
	
	Finally, we show that $F'$ is a free set for $\rel X$, i.e., that the map $h\colon \rel X\to \rel B/{\Theta}$ defined by $h(u)=0$ for every $u\in F'$ and $h(v)=\mathbb Q_{>0}$ for every $v\in X\setminus F'$ is a homomorphism. By~\Cref{claim:proper}, this would give us a contradiction to the minimality of $F$.
	Consider an arbitrary $(u_1,\dots,u_r)\in R^{\rel X}$ for some relation symbol $R$. We let $U=\{u_1,\dots,u_r\}$.
	If $U\cap F'=\emptyset$ then there is nothing to prove, since the tuple $(\mathbb Q_{>0},\dots,\mathbb Q_{>0})$ is in every relation of $\rel A/{\Theta}$.
	So let us assume that there is $u\in U\cap F'$.
	
	For every $w\in F\setminus F'$, let $b\in D_w$ be such that $a<b$ and $(a,b)\in D_{u,w}$, which exists by~\Cref{claim:greater}.
	Since $H$ is a $k$-strategy and $|U|\leq k$, there exists $h_w\in H$ with domain $U$ and such that:
	\begin{itemize}
		\item $(h_w(u_1),\dots,h_w(u_r))\in R^{\rel B}$, since $H$ consists of partial homomorphisms,
		\item $h_w(u)=a < b=h_w(w)$,
		\item $(h(x),h(y))\in H_{x,y}$ for all $x,y\in U$.
	\end{itemize}
	Since $a+H_{u,v}=\{a\}$ for every $v\in F'$, we have $h_w(v)=a$ for every $v\in F'$, and $h_w(v)\geq a$ for every $v\in F$ by~\Cref{claim:lowerbound}.
	
	Let $w_1,w_2,\dots,w_m$ list the elements of $F\setminus F'$.
	Define $g_1=h_{w_1}$ and $g_{i+1}(x)=ll(h_{w_{i+1}}(x), g_i(x))$ for $i\in\{1,\dots,m-1\}$ and $x\in U$.
	We obtain a partial homomorphism $g_m\colon U\to \mathbb Q$ 
	such that $(g_m(u_1),\dots,g_m(u_r))\in R^{\rel B}$, such that $g_m|_{F'}$ is constant, and $g_m(u)<g_m(v)$ for every $u\in F', v\in F\setminus F'$.
	Up to composing $g_m$ with an automorphism of $(\mathbb Q;<)$, we can assume that $g_m|_{F'}$ is constant equal to $0$.
	
	Finally, let $s\in R^{\rel B}$ be a tuple whose existence witnesses the fact that $F$ is a free set, i.e., $s|_F$ is constant equal to $0$ and $s(v)>0$ for every $v\in \{u_1,\dots,u_k\}\setminus F$.
	
	We finally obtain that if one defines $s'(x):=ll(s(x),g_m(x))$ for $x\in U$, then $(s'(u_1),\dots,s'(u_r))\in R^{\rel B}$ and $s'$ is constant equal to $0$ on $F'$ and strictly positive on $\{u_1,\dots,u_r\}\setminus F'$.
	By composing $s'$ with the canonical projection $\mathbb Q_{\geq 0}\to \mathbb Q_{\geq 0}/{\Theta}$, this shows that the constraint under consideration is satisfied in $\rel B/{\Theta}$.
\end{proof}

\begin{proposition}\label{existence-decomposition}
Let $\rel B$ be a temporal structure such that $ll\in\pol(\rel B)$ and let $k\geq 3$ be larger than the arity of $\rel B$.
Let $\rel X$ be an instance of $\CSP(\rel B)$.
If $\rel X$ has a non-trivial well-ordered $k$-strategy $H$, then $(\rel X, H)$ has a decomposition sequence.
\end{proposition}
\begin{proof}
	We prove the result by induction on the size of $\rel X$.
	Let $a$ be the minimal value appearing in any $H_x$ and let $x\in X$ be arbitrary such that $a\in H_x$.
	
	Since $\rel X$ has a well-ordered $k$-strategy with respect to $\rel B$, it is in particular well-ordered singleton arc consistent as witnessed by the non-trivial potato system $(H_x)_{x\in X}$.
	Since $\mathscr C$ contains a semilattice operation,~\Cref{semilattice-sac-free-set} applies and there exists a non-empty free set $F$ where $a\in H_x$ for every $x\in F$. If $F=X$, we are done.
	
	Otherwise, let $F_1\subseteq F$ be an inclusion-minimal free set for $\rel X$.
	We know by~\Cref{ll-min-freeset} that $H_{x,y}$ is a subset of the diagonal for all $x,y\in F_1$.
	Let $G$ be the $k$-strategy from $\rel Y$ to $\rel A$, where $G$ is defined as in~\Cref{projection-strategy} (1) and $\rel Y=(X/F_1)\setminus F_1$ is the contracted projection of $\rel X$ onto $X\setminus F_1$.
	Then by induction hypothesis, $(\rel Y, G)$ has a decomposition sequence $F_2,\dots,F_m$, which means that $F_1,\dots,F_m,$ is a decomposition sequence of $(\rel X, H)$.
\end{proof}

\begin{theorem}\label{ll-semilattice}
	Let $\rel B$ be a temporal structure with relations of arity smaller than $k$ and such that $ll\in\pol(\rel B)$.
	Let $\rel X$ be an instance of\/ $\CSP(\rel B)$.
	If\/ $\rel X$ has a non-trivial well-ordered $k$-strategy, then there exists a homomorphism $\rel X\to\rel B$.
\end{theorem}
\begin{proof}
	Let $H$ be a well-ordered non-trivial $k$-strategy from $\rel X$ to $\rel B$.
	By~\Cref{existence-decomposition}, there exists a decomposition sequence $F_1,\dots,F_m$ for $(\rel X,H)$. 
	By~\Cref{ll-min-freeset}, every $h\in H$ is constant when restricted to any $F_i$.

	We prove that $\rel X$ admits a homomorphism to $\rel B$, namely, that the map $s\colon X\to \mathbb Q$ defined by $s(x)=a$ for $x\in F_a$ is a homomorphism.
	Define $\rel Y_a$ for $a\in\{0,1,\dots,m\}$ recursively by $\rel Y_0=\rel X$ and for $a\geq 1$ by $\rel Y_{a}=(\rel Y_{a-1}/{F_a})\setminus {F_{a}}$.
	Thus, the domain of $\rel Y_a$ is $F_{a+1}\cup\dots\cup F_m$.

	By downward induction on $a\in\{0,\dots,m\}$, we show that $s|_{F_{a+1}\cup\dots\cup F_m}$ is a homomorphism $\rel Y_a\to\rel B$. The case $a=0$ then implies our claim.
	The base case of $a=m$ is clear since $\rel Y_m$ is the empty structure.
	
	Let us assume now that $s|_{F_{a+1}\cup\dots\cup F_m}$ is a homomorphism $\rel Y_a\to\rel B$ for some $a>0$.
	Let $(x_1,\dots,x_r)\in R^{\rel Y_{a-1}}$.
	By definition of a decomposition sequence,  $F_{a}$ is a free set for $\rel Y_{a-1}$, i.e., there exists a homomorphism $h\colon \rel Y_{a-1}\to\rel B/{\Theta}$ with $h^{-1}(0)=F_a$.
	By definition of $\rel B/{\Theta}$, this means that there exists a map $f\colon\{x_1,\dots,x_r\}\to\mathbb Q$ such that $(f(x_1),\dots,f(x_r))\in R^{\rel B}$ and such that $f(y)=0<f(z)$ for all $y\in F_a, z\not\in F_a$.
	Let $x_{i_1},\dots,x_{i_s}$ be an enumeration of the elements of $\{x_1,\dots,x_r\}\setminus F_a$.
	Since $\rel Y_a$ is the projection of $\rel Y_{a-1}/{F_a}$ onto $F_{a+1}\cup\dots\cup F_m$, we have $(x_{i_1},\dots,x_{i_s})\in \tilde R^{\rel Y_a}$, where $\tilde R$ is the symbol corresponding to the relation defined by $\exists_{y\in F_a}\left(R(x_1,\dots,x_r)\land\bigwedge_{z,z'\in F_a} z=z'\right)$ in $\rel B$.
	Since $s|_{F_{a+1}\cup\dots\cup F_m}$ is a homomorphism $\rel Y_a\to\rel B$, there exists a value $b\in\mathbb Q$ such that the map $g\colon \{x_1,\dots,x_r\}\to\mathbb Q$ defined by $g(y)=b$ for $y\in F_a$ and $g(y)=s(y)$ otherwise is such that $(g(x_1),\dots,g(x_r))\in R^{\rel B}$.
	Now it is an observation that the map $s'\colon\{x_1,\dots,x_r\}\to\mathbb Q$ defined by $s'(x):=ll(f(x),g(x))$ induces the same order as $s|_{\{x_1,\dots,x_r\}}$ and is such that $(s'(x_1),\dots,s'(x_r))\in R^{\rel B}$.
	Thus, $(s(x_1),\dots,s(x_r))\in R^{\rel B}$.
	Since this holds for every constraint of $\rel Y_a$,  this implies that $s|_{F_a\cup\dots\cup F_m}$ is a homomorphism $\rel Y_{a-1}\to\mathbb B$.%
\end{proof}

\subsection{Conclusion}
\algotemporal*
\begin{proof}
	Suppose that $\CSP(\rel B)$ is not NP-hard.
    Then it is known that one of the operations $ll$ or $pp$ is a polymorphism of $\rel B$, and that $\mathscr C$ is contains a polymorphism that is either a semilattice, the majority operation, or the minority operation. By~\Cref{maj-implies-semilattice}, it suffices to treat the cases where $\mathscr C$ contains a semilattice operation or the minority operation.
    
    Suppose first that $ll\in\pol(\rel B)$.
    If $\rel X$ has a non-trivial $k$-strategy with respect to $\rel A$, then it has a well-ordered non-trivial $k$-strategy with respect to $\rel B$, and there exists a homomorphism $\rel X\to\rel B$ by~\Cref{ll-semilattice}.
    If $pp\in\pol(\rel B)$ and $\mathscr C$ contains a semilattice, we are likewise done by~\Cref{pp-semilattice}.
  
    Otherwise, $\mathscr C$ contains a minority operation.
    The condition $\sAIP(\rel X,\rel A)$ then implies by~\Cref{pp-maltsev} that there exists a homomorphism $\rel X\to \rel B$.
\end{proof}

\section{PCSP templates from CSP templates}\label{sect:sandwiches}

In this section, we show how to derive half-infinite PCSP templates from $\omega$-categorical templates and transfer properties between the corresponding templates.

Our goal is two-fold.
On the one hand, we prove that for every CSP template in the Bodirsky-Pinsker class, there exists a half-infinite PCSP template such that the corresponding problems are equivalent under logspace reductions.
We moreover show some correspondences between the two problems in terms of their solvability by certain algorithmic techniques that are standard in constraint satisfaction.
On the other hand, we show some correspondences between the respective sets of polymorphisms of these templates.

The following is a known construction allowing to turn arbitrary relational structures into structures in an at most binary signature via a pp-power.
\begin{definition}
Let $\rel B$ be a structure in a finite relational signature, $m\geq 1$ a natural number that is at least the maximal arity of the relations of $\rel B$.
We define $\mpow{\rel B}$ to be the structure whose domain is $B^m$ and with the following relations:
\begin{itemize}
    \item for every $\ell$-ary relation $R$ of $\rel B$, and every $i\colon[\ell]\to[m]$, $\mpow{\rel B}$ has a relation $R_i$ of arity $1$ containing all elements $a\in B^m$ such that $(a_{i(1)},\dots,a_{i(\ell)})\in R^{\rel B}$.
    \item for every $1\leq\ell\leq m$ and every two functions $i,j\colon [\ell]\to[m]$, $\mpow{\rel B}$ has a relation $eq_{i,j}$ of arity $2$ containing all pairs $(a,b)\in (B^{m})^2$ such that $(a_{i(1)},\dots,a_{i(\ell)})=(b_{j(1)},\dots,b_{j(\ell)})$.
\end{itemize}
\end{definition}
It is clear that $\rel B$ and $\mpow{\rel B}$ pp-construct each other, and therefore $\CSP(\rel B)$ and $\CSP(\mpow{\rel B})$ are polynomial-time equivalent.

By definition, if $\rel A\to\rel B$, one always has homomorphisms $\mpow{\rel A}\to\mpow{\rel B}\to \mpow{\rel B}/{\group}$, and therefore there is a trivial reduction from $\PCSP(\mpow{\rel A},\mpow{\rel B}/{\group})$ to $\CSP(\mpow{\rel B})$, which itself reduces to $\CSP(\rel B)$ as already mentioned.

\vspace{-1em}

\subsection{More about Polymorphisms}

A \emph{minor identity} is a formal expression $f^\sigma\approx g^\tau$, for some operation symbols $f,g$. A \emph{minor condition} is a set $\Sigma$ of minor identities. We say that $\Sigma$ is satisfied in $\pol(\rel A,\rel B)$ if for every symbol $f$ appearing in $\Sigma$, one can find an operation $\xi(f)\in\pol(\rel A,\rel B)$ of the corresponding arity and such that $\xi(f)^\sigma=\xi(g)^\tau$ holds for every identity $f^\sigma\approx g^\tau$ in $\Sigma$.

Given $\mathcal U\subseteq\End(\rel B)$, we say that a minor condition $\Sigma$ is satisfied in $\pol(\rel A,\rel B)$ \emph{modulo $\mathcal U$} if there exists an assignment $\xi$ of the symbols from $\Sigma$ to elements of $\pol(\rel A,\rel B)$ of the corresponding arity such that for every identity $f^\sigma\approx g^\tau$ in $\Sigma$, there exist $u,v\in\mathcal U$ such that $u\circ \xi(f)^\sigma=v\circ\xi(g)^\tau$ holds.
In the following, we use the symbol $\trivgrp$ to denote the set consisting solely of the identity function on the relevant set.

Given $\mathscr U\subseteq\End(\rel B)$ and $\mathscr V\subseteq\End(\rel D)$, we say that an arity-preserving map $\xi\colon\pol(\rel A,\rel B)\to\pol(\rel C,\rel D)$ is a \emph{$(\mathscr U,\mathscr V)$-minion homomorphism} if whenever the equality 
\[ u\circ f^\sigma = v \circ g^\tau \]
holds for some $f,g\in\pol(\rel A, \rel B), u,v\in\mathscr U$, one has that the equality
\[ u'\circ \xi(f)^\sigma = v'\circ \xi(g)^\tau\]
holds for some $u',v'\in\mathscr V$.
We say that $\xi$ is a \emph{local} $(\mathscr U,\mathscr V)$-minion homomorphism if for every finite $T\subseteq C$, there exists a finite $S\subseteq A$ such that the implication
\[ \begin{array}{rl} &\exists u,v\in\mathscr U \left(u\circ f^\sigma|_S = v\circ g^\tau|_S\right) \\\Longrightarrow & \exists u,v\in\mathscr V \left(u\circ \xi(f)^\sigma|_T = v\circ \xi(g)^\tau|_T\right)\end{array}\]
holds for all $f,g\in\pol(\rel A,\rel B)$.
\emph{Minion homomorphisms} as defined in~\cite{wonderland} correspond exactly to $(\trivgrp,\trivgrp)$-minion homomorphisms.

We say that a sequence $f_i\colon A^n\to B$ of functions of the same arity converges to $g\colon A^n\to B$ if for every finite subset $S$ of $B$, there exists $i_0\in\mathbb N$ such that for all $i\geq i_0$ we get $f_i|_S=g|_S$.
Seeing $\aut(\rel B)$ as a subset of $\pol(\rel B)$, we let $\overline{\aut(\rel B)}$ be its closure (the set of limits) in this topology.

\begin{lemma}[Standard compactness argument, folklore]\label{sca}
Let $f_i\colon A^n\to B$ be an arbitrary sequence of functions and $\group$ is an oligomorphic group of permutations on $\rel B$.
There exists $g\colon A^n\to B$ such that for every finite $S\subseteq A$ and for infinitely many $i$, $g|_S$ and $f_i|_S$ are in the same orbit under $\group$.
\end{lemma}

If $(\rel C,\rel D)$ is pp-constructible in $(\rel A,\rel B)$, then there exists a local minion homomorphism $\pol(\rel A,\rel B)\to\pol(\rel C,\rel D)$.
Theorem 4.12 in~\cite{BBKO} provides a reciprocal statement in case $\rel A$ and $\rel C$ are finite structures: if there exists a minion homomorphism $\pol(\rel A,\rel B)\to\pol(\rel C,\rel D)$, then $(\rel C,\rel D)$ has a pp-construction in $(\rel A,\rel B)$.

We record here a slight generalization of this fact that also generalizes the results concerning (local) minion homomorphisms and pp-constructions in the context of clones and CSP templates from~\cite{wonderland} to the case of minions and PCSP templates.
We consider it to be folklore; a proof can be given in the language of~\cite{BBKO}, but it can just as well be proved as in~\cite{wonderland}.
\begin{restatable}{proposition}{constructionshomo}\label{prop:pp-constructions-from-homo}
    Let $\rel B$ be an $\omega$-categorical structure.
    Let $\xi\colon\pol(\rel A,\rel B)\to\pol(\rel C,\rel D)$ be a local minion homomorphism.
    Then $(\rel A,\rel B)$ pp-constructs every $(\rel S,\rel D)$ where $\rel S$ is a finite substructure of $\rel C$.
\end{restatable}

\vspace{-1.5em}
\subsection{Categorical properties of $\mpow{}$}\label{sect:category}

\begin{restatable}{proposition}{directedsystem}\label{directedsystem}
	Let $(\rel A,\rel B)$ be a PCSP template and let $\group$ be a subgroup of $\aut(\rel B)$. Let $1\leq m\leq n$. The following hold:
	\begin{itemize}
	\item There exists a local $(\group,\trivgrp)$-minion homomorphism $\xi^\infty_m\colon\pol(\rel A,\rel B)\to\pol(\mpow{\rel A},\mpow{\rel B}/{\group})$.
	\item There exists a local $(\trivgrp,\trivgrp)$-minion homomorphism $\xi^n_{m}\colon\pol(\mpow[n]{\rel A},\mpow[n]{\rel B}/{\group})\to\pol(\mpow{\rel A},\mpow{\rel B}/{\group})$.
	\end{itemize}
	Moreover, $\xi^m_{\ell}\circ \xi^n_m = \xi^n_\ell$ holds for all $\ell\leq m\leq n$ with $\ell,m,n\in\mathbb N\cup\{\infty\}$.
\end{restatable}

We thus have that the family of minions $\pol(\mpow{\rel A},\mpow{\rel B}/{\group})$ with the maps $(\xi^n_m)_{n\geq m}$ is an inverse system, and we have a cone from $\pol(\rel A,\rel B)$ to this inverse system, see~\Cref{fig:proj-lim}.
It is not true that $\pol(\rel A,\rel B)$ is the projective limit of the system, however it is ``almost'' one if $\group$ is oligomorphic (\Cref{proj-limit}).

The signature of $\mpow{\rel B}$ only depends on the signature $\tau$ of $\rel B$ and is denoted by $\mpow{\tau}$.
Thus, $\mpow{}$ defines a map from the category $\struct[\tau]$ of $\tau$-structures to the category $\struct[\mpow{\tau}]$ of $\mpow{\tau}$-structures.
We show that this map has an adjoint $\mlow{}$.
This adjunction in $\struct$ helps us to prove that $\mpow{}$ preserves products up to isomorphism (\Cref{lem:fun-fact}), and to prove some adjunction-like results in certain categories containing $\struct$ (\Cref{graded-adjoint}), allowing us to transfer results about the power of certain hierarchies of algorithms (\Cref{relating-hierarchies}).

Given a structure $\rel X$ in the signature $\mpow{\tau}$, define $\mlow{\rel X}$ to be the $\tau$-structure obtained by replacing every element $a$ by $m$ elements $a_1,\dots,a_m$ satisfying appropriate constraints and identifying elements according to the $eq$ relations. Formally, $\mlow{\rel X}$ is defined as follows:
\begin{itemize}
    \item Let $Y=X\times [m]$.
    \item Let $((a,i(1)),\dots,(a,i(\ell))\in R^{\rel Y}$ iff $a\in R^{\rel X}_i$.
    \item Define $(a,p)\leftrightarrow (b,q)$ if $(a,b)\in eq^{\rel X}_{i,j}$ for some $i,j\colon [\ell]\to [m]$, and if $i^{-1}(p)\cap j^{-1}(q)\neq\emptyset$.
    \item Let $\sim^{\rel X}$ be the smallest equivalence relation containing $\leftrightarrow$ and define $\mlow{\rel X}=\rel Y/{\sim^{\rel X}}$.
\end{itemize}

Viewing $\struct[\tau]$ and $\struct[\mpow{\tau}]$ as categories whose arrows are homomorphisms of structures, we get that $\mpow{}$ and $\mlow{}$ are adjoint functors.
This is a particular example of functors arising from a \emph{Pultr template}, see e.g.~\cite{DalmauKrokhinOprsal}, although in general Pultr templates do not give rise to adjunctions.%

In the following, if $\cate{C}$ is a category and $A,B\in\cate{C}$ are two objects, then $\cate{C}(A,B)$ is the collection of arrows in $\cate C$ with domain $A$ and codomain $B$.

\begin{restatable}{lemma}{functoriality}\label{functoriality}
Let $m$ be larger than the arity of the relations of the structures under consideration.
Then:
\begin{itemize}
\item $\mpow{}$ is a fully faithful functor,
\item $\mlow{}$ is a faithful functor,
\item $\mlow{\mpow{\rel X}}\simeq \rel X$ for all structures $\rel X$,
\item $\mlow{}$ and $\mpow{}$ form an adjoint pair, i.e., there is a bijection $\struct(\rel X,\mpow{\rel B}) \to \struct(\mlow{\rel X},\rel B)$  natural in $\rel X$ and $\rel B$.
\end{itemize}
\end{restatable}

\begin{corollary}\label{lem:fun-fact}
For every structure $\rel A$ and all integers $n,m$, the structures $\mpow{(\rel A^n)}$ and $(\mpow{\rel A})^n$ are isomorphic.
\end{corollary}
\begin{proof}
    This follows from~\Cref{functoriality} and the fact that functors with a left adjoint preserve products.
\end{proof}

For the purpose of the first item in~\Cref{bw-not-finitely-tractable}, we are interested in comparing the power of certain algorithmic hierarchies for solving $\PCSP(\rel A,\rel B)$ and the related problem $\PCSP(\mpow{\rel A},\mpow{\rel B}/{\group})$.
In particular, we focus on some hierarchies known as \emph{minion tests}~\cite{CiardoZivnyHierarchyTensor}
and \emph{consistency reductions}~\cite{ConsistencyReduction}.
Such a consistency reduction is parametrized by a dimension $k\in\mathbb N$ and a minion $\minion{}$.
Given a structure $\rel X$ as input to $\PCSP(\rel A,\rel B)$, the reduction first computes a $k$-strategy $H$ from $\rel X$ to $\rel A$.
In a second step, the reduction asks to find, for each small set $K\subseteq X$, an element $\xi_K\in\minion{H_K}$ such that for every $K\subseteq L\subseteq X$ of size at most $k$, one has $(\xi_L)^\sigma=\xi_H$, $\sigma$ being the inclusion map $K\hookrightarrow L$.
Such a reduction is always complete for a non-empty minion, i.e., if there exists a homomorphism $\rel X\to\rel A$ then a $k$-strategy $H$ and corresponding elements $\xi_K\in\minion{H_K}$ do exist.
The soundness of that reduction, stating that the existence of such witnesses imply the existence of a homomorphism $\rel X\to\rel B$, might fail.

To compare the power of such reductions for solving $\PCSP(\rel A,\rel B)$ and $\PCSP(\mpow{\rel A},\mpow{\rel B}/{\group})$, we study properties of morphisms in a category containing $\cate{Rel}$ and whose arrows correspond to witnesses of acceptance by the reduction.

Let $\Systems{k}$ be the category of relational structures whose arrows $\Systems{k}(\rel X,\rel Y)$ correspond to non-trivial $k$-strategies from $\rel X$ to $\rel Y$.
Let \minion{} be a minion. Consider the subclass $\MSystems{k}$ of $\Systems{k}$ whose objects are relational structures, and whose arrows are pairs $(H,\xi)\colon \rel A\to\rel B$ such that:
\begin{itemize}
	\item $H$ is a non-trivial $k$-strategy from $\rel A$ to $\rel B$,
	\item $\xi$ is a map from $\binom{A}{\leq k}$ to $\minion$ such that for every $K\in\binom{A}{\leq k}$, $\xi_K=\xi(K) \in\minion H_K$,
	\item for every $K\subseteq L\in\binom{A}{\leq k}$, we have $(\xi_L)^{\sigma}=\xi_K$, where $\sigma$ is the natural restriction map $H_L\to H_K$.
\end{itemize}

Due to the aforementioned  completeness of the reduction, $\cate{Rel}$ is always a subcategory of $\MSystems{k}$ (if \minion{} is non-empty).
For $\minion$ being the clone of projections on a 2-element set, the resulting category coincides with $\cate{Rel}$, for all $k$.
In general, if $\minion=\pol(\rel C,\rel D)$, then there exists an arrow in $\MSystems{k}(\rel X,\rel A)$ iff $\rel X$ is accepted by the $k$-consistency reduction from $\PCSP(\rel A, \underline{\hspace{2mm}})$ to $\PCSP(\rel C,\rel D)$ as defined in~\cite{ConsistencyReduction}.

By~\Cref{functoriality} and the observation that $\MSystems{k}$ coincides with $\struct{}$ when \minion{} is the clone of projection and $k$ is large enough, we obtain that in this case $\mlow{}$ and $\mpow{}$ form a pair of adjoint functors on $\MSystems{k}$.
This does not necessarily hold for general linear minions,  but we do keep some properties related to adjunction.

\begin{restatable}{lemma}{gradedadjoint}\label{graded-adjoint}
	Let \minion{} be a linear minion.
	The following hold for every $k,m\geq 1$:
	\begin{itemize}
		\item there is a natural map $\MSystems{m\cdot k}(\mlow{\rel X},\rel B)\to\MSystems{k}(\rel X,\mpow{\rel B})$,
		\item there is a natural map $\MSystems{k}(\rel X,\mpow{\rel B})\to\MSystems{k}(\mlow{\rel X},\rel B)$.
	\end{itemize}
\end{restatable}

Note that $\MSystems{k}$ might not be a category in general, as the composition of arrows is not necessarily well defined.
It is possible, however, to give a sufficient condition for $\MSystems{k}$ to be a category.
Following~\cite{CiardoZivnyHierarchyTensor}, a minion \minion{} is called \emph{linear} if there exists $d\geq 1$ and a semiring $R$ such that the elements of $\minion X$ are $(X\times [d])$-matrices with entries in $R$, and where for $\sigma\colon X\to Y$ and $M\in\minion X$, then $M^\sigma = P_\sigma M$, with $P_\sigma$ the $(Y\times X)$-matrix where $P_{\sigma}(y,x)=1$ if $\sigma(x)=y$.

Let $\minion$ be a linear minion.
Let $\mathcal F\subseteq B^A$ and $\mathcal G\subseteq C^B$ be finite sets of maps,
and let $\xi\in\minion{\mathcal F}$ and $\zeta\in\minion{\mathcal G}$.
Let $\mathcal H=\mathcal G\circ\mathcal F=\{g\circ f\mid g\in\mathcal G, f\in\mathcal F\}$.
Define $\zeta*\xi$ as the matrix with rows indexed by $\mathcal H$, and whose row indexed by $h\colon A\to C$ equals $\sum_{f\in \mathcal F} \sum_{g\in \mathcal G, h=g\circ f} \xi(f)\odot \zeta(g)$, where the operator $\odot$ denotes pointwise multiplication of vectors of length $d$.
We say that a linear minion \minion{} is \emph{good-for-composition} if for all finite %
sets of maps $\mathcal F,\mathcal G$ and all $\xi\in\minion{\mathcal F},\zeta\in\minion{\mathcal G}$, we have $\zeta*\xi\in\minion{(\mathcal G\circ\mathcal F)}$.

\begin{restatable}{proposition}{category}\label{category}
	Let \minion{} be good-for-composition, and let $k\geq 1$.
	Then $\MSystems{k}$ is a category.
\end{restatable}

In particular, this allows us to relate the power of hierarchies of consistency reductions for templates of the form $(\mpow{\rel A},\mpow{\rel B})$ and $(\rel A,\rel B)$.
\begin{corollary}\label{relating-hierarchies}
Let \minion{} be  good-for-composition.
Let $(\rel A,\rel B)$ be a PCSP template, let $k\geq 1$ and let $m$ be greater than the arity of the relations of $\rel A$.
The following hold:
\begin{itemize}
	\item If\/ $\PCSP(\mpow{\rel A},\mpow{\rel B})$ is solvable by the $k$-consistency reduction to $\minion$, then $\PCSP(\rel A,\rel B)$ is solvable by the $mk$-consistency reduction to \minion.
	\item If\/ $\PCSP(\rel A,\rel B)$ is solvable by the $k$-consistency reduction to \minion, then so is $\PCSP(\mpow{\rel A},\mpow{\rel B})$.
\end{itemize}
\end{corollary}
\begin{proof}
Suppose that $\PCSP(\mpow{\rel A},\mpow{\rel B})$ is solved by the $k$-consistency reduction to \minion{}. Assume that there exists an arrow in $\MSystems{m\cdot k}(\rel X,\rel A)$, for some instance $\rel X$ of $\PCSP(\rel A,\rel B)$.
By~\Cref{functoriality}, we  have $\mlow{\mpow{\rel X}}\to\rel X$, and therefore $\MSystems{m\cdot k}(\mlow{\mpow{\rel X}},\rel X)$ is non-empty.
If \minion{} is good-for-composition, then there exists an arrow in $\MSystems{m\cdot k}(\mlow{\mpow{\rel X},\rel A})$ by~\Cref{category}.
By~\Cref{graded-adjoint}, we obtain an arrow in $\MSystems{k}(\mpow{\rel X},\mpow{\rel A})$.
Since $(\mpow{\rel A},\mpow{\rel B})$ is solved by the $k$-consistency reduction to \minion{}, we get that there exists a homomorphism $\mpow{\rel X}\to\mpow{\rel B}$, and thus $\rel X\to\rel B$ since $\mpow{}$ is full (\Cref{functoriality}).

Similarly, assume that $\PCSP(\rel A,\rel B)$ is solved by the $k$-consistency reduction to \minion{}, and let $\rel X$ be an instance of $\PCSP(\mpow{\rel A},\mpow{\rel B})$ such that $\MSystems{k}(\rel X,\mpow{\rel A})$ is non-empty.
Then $\MSystems{k}(\mlow{\rel X},\rel A)$ is non-empty by~\Cref{graded-adjoint}, which implies that $\mlow{\rel X}\to\rel B$.
Thus, $\rel X\to\mpow{\rel B}$ by~\Cref{functoriality}.
\end{proof}

\vspace{-1em}

\subsection{Proof of~\Cref{intro-pcsp-csp}}

Let $\Gamma$ be any reduction from $\PCSP(\rel A,\rel B)$ to $\PCSP(\rel C,\rel D)$.
If $\pi\colon \hom(\underline{\hspace{2mm}},\rel B)\to\hom(\Gamma(\underline{\hspace{2mm}}),\rel D)$ is a function, we say that $\Gamma$  is \emph{$\pi$-full} if for every $\rel X$ and every homomorphism $h\colon\Gamma\rel X\to\rel D$, there exists a homomorphism $g\colon\rel X\to\rel B$ such that $\pi(g) = h$.
In the following, we naturally study the operator $\Gamma_m$ seen as a reduction and the associated map $\pi\colon\hom(\underline{\hspace{2mm}},\rel B)\to\hom(\mpow{\underline{\hspace{2mm}}},\mpow{\rel B}/{\group})$ that sends $f$ to $(a_1,\dots,a_m)\mapsto (f(a_1),\dots,f(a_m))/{\group}$.

The following is closely related to Theorem 3 in~\cite{ReductionFinite}.
We note that putting aside the algorithmic consequences, the result also applies when considering infinite instances; we will make use of this in~\Cref{lem:lift-poly}.
\begin{restatable}{lemma}{reductionfinite}\label{lem:reduction-csp-pcsp}
    Let $m\geq 3$.
    Let $(\rel A,\rel B)$ be a PCSP template, where $\rel B$ is a first-order reduct of a homogeneous finitely bounded structure $\rel B^+$ whose obstructions have size at most $m$, and whose relations have arity at most $m-1$.
    Then $\mpow{}$ is a full, sound and complete reduction from $\PCSP(\rel A,\rel B)$ to $\PCSP(\mpow{\rel A},\mpow{\rel B}/{\aut(\rel B^+)})$.
\end{restatable}
An immediate consequence of the previous statement is~\Cref{intro-pcsp-csp} in the introduction, namely that the problems in the scope of the Bodirsky-Pinsker conjecture form a subset of ``half-infinite'' promise CSPs, where the left-hand side is itself in the scope of the Bodirsky-Pinsker conjecture and the right-hand side is finite.
 This is a consequence of the following statement in the special case where $\rel A=\rel B$.
\begin{corollary}\label{cor:csp-pcsp-equivalent}
    Let $m\geq 3$.
    Let $(\rel A,\rel B)$ be as in~\Cref{lem:reduction-csp-pcsp}. %
    Then $\PCSP(\rel A,\rel B)$ is polynomial-time equivalent to $\PCSP(\mpow{\rel A},\mpow{\rel B}/{\aut(\rel B^+)})$.
\end{corollary}
\begin{proof}
    One reduction is given by~\Cref{lem:reduction-csp-pcsp}, and $\PCSP(\mpow{\rel A},\mpow{\rel B}/{\aut(\rel B^+)})$ reduces to $\PCSP(\mpow{\rel A},\mpow{\rel B})$ by the trivial reduction since $\mpow{\rel A}\to\mpow{\rel B}\to\mpow{\rel B}/{\aut(\rel B^+)}$.
    Finally, $\PCSP(\rel A,\rel B)$ and  $\PCSP(\mpow{\rel A}, \mpow{\rel B})$ are equivalent, since the templates pp-construct each other.
\end{proof}

\subsection{Identities modulo unaries}\label{sect:algebra}

In this section, we relate the satisfiability of minor identities in $\pol(\rel A,\rel B)$ modulo $\group$ and satisfiability in $\pol(\mpow{\rel A},\mpow{\rel B}/{\group})$.
We first give some transfer principles in the case where $\group\leq\aut(\rel B)$ is known to be the automorphism group of a first-order reduct of a finitely bounded homogeneous structure $\rel B^+$.
In that case, the identities satisfied in $\pol(\rel A,\rel B)$ modulo $\overline{\group}$ are completely captured by $\pol(\mpow{\rel A},\mpow{\rel B}/{\group})$ for some fixed $m$ large enough. This is established in the following two statements.
\begin{proposition}\label{lem:lift-poly}
    Let $m\geq 3$.
    Let $(\rel A,\rel B)$ be a PCSP template where $\rel B$ is a first-order reduct of a homogeneous finitely bounded structure $\rel B^+$ whose bounds have size at most $m$, and whose relations have arity at most $m-1$.
    For every $f\in\pol(\mpow{\rel A},\mpow{\rel B}/{\aut(\rel B^+)})$ of arity $n$, there exists $f^*\in\pol(\rel A,\rel B)$ such that $\xi_m^\infty(f^*)=f$.
\end{proposition}
\begin{proof}
    Let $\rel X=\rel A^n$.
    By~\Cref{lem:fun-fact}, the structures $\mpow{\rel X}$ and $(\mpow{\rel A})^n$ are isomorphic.
    Thus, $f$ is a homomorphism $\mpow{\rel X}\to \mpow{\rel B}/{\aut(\rel B^+)}$.
    By~\Cref{lem:reduction-csp-pcsp}, there exists $f^*\colon \rel X\to \rel B$ such that $\pi(f^*)=f$, where $\pi$ maps $f^*$ to $(a_1,\dots,a_n)\mapsto f^*(a_1,\dots,a_n)/{\aut(\rel B^+)}$.
    Thus, $\xi_m^\infty(f^*)=f$.
\end{proof}

\begin{corollary}\label{cor:local-h1-finbounded}
	Let $m\geq 3$.
    Let $(\rel A,\rel B)$ be a PCSP template where $\rel B$ is a first-order reduct of a homogeneous finitely bounded structure $\rel B^+$ whose bounds have size at most $m$, and whose relations have arity at most $m-1$.
	There exists a local $(\trivgrp,\aut(\rel B^+))$-minion homomorphism $\pol(\mpow{\rel A},\mpow{\rel B}/{\aut(\rel B^+)})\to\pol(\rel A,\rel B)$.
\end{corollary}
\begin{proof}
    Given an $n$-ary $f\in \pol(\mpow{\rel A}, \mpow{\rel B}/{\aut(\rel B^+)})$, let $f^*\in\pol(\rel A,\rel B)$ be as given by~\Cref{lem:lift-poly}.

    We show that $f\mapsto f^*$ is a local $(\trivgrp,\aut(\rel B)^+)$-minion homomorphism.
    Let $T$ be a finite subset of $A$.
    Let $S=T^m$.
    Suppose that $g|_S=f^\sigma|_S$ for some $f,g\in\pol(\mpow{\rel A}, \mpow{\rel B}/{\aut(\rel B^+)})$, where $\sigma\colon [p]\to[q]$.
    
  Let $a^1,\dots,a^q$ be $m$-tuples of elements of $T$.
  By~\Cref{lem:lift-poly}, we have $g^*(a^1,\dots,a^q)\in g(a^1,\dots,a^q) = f^{\sigma}(a^1,\dots,a^q)=f(a^{\sigma(1)},\dots,a^{\sigma(p)})$
  and $(f^*)^{\sigma}(a^1,\dots,a^q) = f^*(a^{\sigma(1)},\dots,a^{\sigma(p)})\in f(a^{\sigma(1)},\dots,a^{\sigma(p)})$,
  thus $g^*(a^1,\dots,a^q)$ and $(f^*)^\sigma(a^1,\dots,a^q)$ are in the same orbit under $\aut(\rel B^+)$.
  Since $\rel B^+$ is homogeneous in a language of arity $m$, it follows that $g^*(a^1,\dots,a^q)$ is in the same orbit as $(f^*)^{\sigma}(a^1,\dots,a^q)$ for tuples of arbitrary length.
  We therefore obtain the existence of $\alpha\in\aut(\rel B^+)$ such that
    $\alpha g^*|_T=(f^*)^\sigma|_T$ holds.
\end{proof}

In the more general case that $\group\leq\aut(\rel B)$ is an arbitrary oligomorphic subgroup, the situation is slightly more complex and the identities satisfied in $\pol(\rel A,\rel B)$ modulo $\overline{\group}$ can be read from ``compatible'' satisfaction in the minions $\pol(\mpow{\rel A},\mpow{\rel B}/{\group})$.
\begin{restatable}{proposition}{liftoligo}\label{prop:lift-poly-oligo}
    Let $(\rel A,\rel B)$ be a PCSP template, and let $\group$ be an oligomorphic subgroup of $\aut(\rel B)$.
    Let $(f_m)_{m\geq 1}$ be such that $f_m\in\pol(\mpow{\rel A},\mpow{\rel B}/{\group})$ and such that $\xi^n_{m}(f_{n})=f_m$ for all $n\geq m\geq 1$.
    There exists $f^*\in\pol(\rel A,\rel B)$ such that $\xi^\infty_m(f^*)=f_m$ for all $m\geq 1$.
\end{restatable}
\begin{proof}
Let $n$ be the arity of the operations $f_m$.
Suppose first that $\rel A$ is finite.
Let $m$ be at least as large as $|A|^n$ and as the arity of the relations of $\rel A$ and $\rel B$.
    By~\Cref{lem:fun-fact}, $f_m\colon (\mpow{\rel A})^n\to \mpow{\rel B}/{\group}$ can be seen as a homomorphism $\tilde f_m\colon\mpow{(\rel A^n)}\to\mpow{\rel B}/{\group}$. Consider an arbitrary element $(a^1,\dots,a^m)\in \mpow{(\rel A^n)}$ where $a^1,\dots,a^m$ enumerate all of $A^n$. %
    Let $(b_1,\dots,b_m)$ be any element from $\tilde f_m(a^1,\dots,a^m)$.
    Define $f^*(a^i)=b_i$ for all $i\in\{1,\dots,m\}$.
    
    We first prove that this is well-defined.
    Suppose  that $a^i=a^j$. This implies that $((a^1,\dots,a^m),(a^1,\dots,a^m))$ is in the relation $eq_{(i,i),(i,j)}$ of $\mpow{(\rel A^n)}$.
    It follows that $\pr{i,i}(\tilde f_m(a^1,\dots,a^m))$ and $\pr{i,j}(\tilde f_m(a^1,\dots,a^m))$ are the same $\group$-orbit of pairs, i.e., $(b_i,b_i)$ and $(b_i,b_j)$ are in the same orbit under $\group$, which is only possible if $b_i=b_j$.
    We also obtain that for every $m'\geq1$ and all $m'$-tuples $c^1,\dots,c^n$, we have $f^*(c^1,\dots,c^n)\in f_{m'}(c^1,\dots,c^n)$ and therefore $\xi^\infty_{m'}(f^*)=f_{m'}$.
    
	We now prove that $f^*$ is a polymorphism of $(\rel A,\rel B)$.
	Let $c^1,\dots,c^n\in R^{\rel A}$, for some relation symbol $R$ of arity $r$.
	Let $c'^i$ be the $m$-tuple obtained by padding $c^i$ with $m-r$ arbitrary elements of $A$.
	We see $c'^1,\dots,c'^n$ as the columns of a matrix of size  $m\times n$.
	For each $p\in\{1,\dots,m\}$, there exists $i_p\in\{1,\dots,m\}$ such that the $p$th row of this matrix is one of the vectors $a^{i_p}$.
	The result of applying $f^*$ componentwise to $c'^1,\dots,c'^n$ is then $(b_{i_1},\dots,b_{i_m})$, so we want to prove that $(b_{i_1},\dots,b_{i_r})\in R^{\rel B}$.
	Note that $(a^{i_1},\dots,a^{i_m})\in R_{1,\dots,r}^{\mpow{(\rel A^n)}}$.
	Therefore, we get that $\tilde f_m(a^{i_1},\dots,a^{i_m})\in R_{1,\dots,r}^{\mpow{\rel B}/{\group}}$, i.e., some $m$-tuple in $\tilde f_m(a^{i_1},\dots,a^{i_m})$ is an element of $R_{1,\dots,r}^{\mpow{\rel B}}$.
	Since $\group$ is a subgroup of $\aut(\rel B)$, \emph{all} the tuples in $\tilde f_m(a^{i_1},\dots,a^{i_m})$ belong to $R^{\mpow{\rel B}}_{1,\dots,r}$.
	Moreover, since $((a^{i_1},\dots,a^{i_m}),(a^1,\dots,a^m))$ is in the relation $eq_{(1,\dots,m),(i_1,\dots,i_m)}$ of $\mpow{(\rel A^n)}$, we get that
	the orbit $\tilde f_m(a^{i_1},\dots,a^{i_m})$ coincides with the projection of the orbit $\tilde f_m(a^1,\dots,a^m)$ onto the coordinates $i_1,\dots,i_m$.
	This is exactly the orbit of $(b_{i_1},\dots,b_{i_m})$ by definition, and therefore $(b_{i_1},\dots,b_{i_m})\in R^{\rel B}$.
	
	Finally, we turn to the case where $\rel A$ is countably infinite.
	Let $\rel A_1,\rel A_2,\dots$ be finite substructures of $\rel A$ such that $\rel A_i\subseteq\rel A_{i+1}$ for all $i\geq 1$ and $\bigcup_{i\geq 1} \rel A_i=\rel A$.
	For all $i\geq 1$, the family $(f_m)_{m\geq 1}$ restricts to a family $(g^i_m)_{m\geq 1}$ of maps in $\pol(\mpow{\rel A_i},\mpow{\rel B}/{\group})$ satisfying the condition of the statement, and therefore we obtain a $g_i^*\in\pol(\rel A_i,\rel B)$ with $\xi^\infty_m(g_i^*)=g^i_m$ for all $m\geq 1$.
	By~\Cref{sca}, one obtains an $f^*$ such that $\xi^\infty_m(f^*)=f_m$.
\end{proof}

It follows that $\pol(\rel A,\rel B)$ can loosely be thought of as a projective limit of the finite minions $\pol(\mpow{\rel A},\mpow{\rel B}/{\group})$.
\begin{restatable}{theorem}{projlimit}\label{proj-limit}
	Let $(\rel A,\rel B)$ be a PCSP template, and let $\group$ be an oligomorphic subgroup of $\aut(\rel B)$.
	Let $\minion{}$ be a minion and $\zeta_m\colon \minion{}\to\pol(\mpow{\rel A},\mpow{\rel B}/{\group})$ be a system of minion homomorphisms such that $\xi^n_m\circ\zeta_n = \zeta_m$.
	Then there exists a $(\trivgrp,\overline{\group})$-minion homomorphism $\zeta^*\colon\minion{}\to\pol(\rel A,\rel B)$ such that $\xi^\infty_m\circ \zeta^* = \zeta_m$ holds for all $m$.
\end{restatable}

\begin{figure}
    \centering
    \begin{tikzcd}
	& & \vdots\arrow[d,"\xi^{m+2}_{m+1}"]\\
	& & \pol(\mpow[m+1]{\rel A},\mpow[m+1]{\rel B}/{\group}) \arrow[d,"\xi^{m+1}_m"]\\
	\minion{}\arrow[urr,bend left,end anchor=west,"\zeta_{m+1}"]\arrow[drr,bend right,end anchor=west,"\zeta_{m-1}"] \arrow[r, dotted, "{\zeta^*}"] & \pol(\rel A,\rel B) \arrow[ur,"\xi^\infty_{m+1}"] \arrow[r,"\xi^\infty_{m}"]\arrow[dr,"\xi^\infty_{m-1}"] & \pol(\mpow{\rel A},\mpow{\rel B}/{\group})\arrow[d,"\xi^m_{m-1}"]\\
	& & \pol(\mpow[m-1]{\rel A},\mpow[m-1]{\rel B}/{\group})\arrow[d,"\xi^{m-1}_{m-2}"]\\
	& & \vdots
\end{tikzcd}
    \caption{Diagram representing the various arrows in~\Cref{directedsystem,proj-limit}. Solid arrows are local minion homomorphisms, and the dotted arrow is a $(\trivgrp,\overline{\group})$-homomorphism.}
    \label{fig:proj-lim}
\end{figure}

It follows in particular that if a minor condition is satisfiable in $\pol(\mpow{\rel A},\mpow{\rel B}/{\group})$ for all $m\geq 1$, then it is satisfiable in $\pol(\rel A,\rel B)$ modulo $\overline{\group}$: take for this $\minion{}$ to be the minion generated by the operations in the minor condition.
Here, we can renounce the compatibility conditions $\xi^n_m\circ\zeta_n=\zeta_m$, since we are only interested in the existence of an arbitrary $(\trivgrp,\overline{\group})$-minion homomorphism $\zeta^*$.

  \vspace{-1em}
\subsection{Important special cases}

We briefly recall the following minor conditions and their importance in constraint satisfaction.
The \emph{Ol\v{s}\'ak condition} is the minor condition
$ f(x,x,y,y,y,x)\approx f(x,y,x,y,x,y) \approx f(y,x,x,x,y,y)$.
The 6-ary Siggers condition is the minor condition
$s(x,y,x,z,y,z) \approx s(y,x,z,x,z,y)$, the 4-ary Siggers condition is the minor condition
		 $s(x,y,z,x) \approx s(y,z,x,z)$.
		For $k\geq 3$,  $\WNU(k)$ is the minor condition with a $k$-ary symbol $f$ and the identities
		\[ f(x,y,\dots,y) \approx f(y,x,y,\dots,y)\approx \dots\approx f(y,\dots,y,x).\]

It is known that for a finite structure $\rel A$, $\pol(\rel A)$ satisfies one of the given conditions if, and only if, it satisfies all of them~\cite{Olsak,Siggers,KVVW,MarotiMcKenzie} (for the case of $\WNU$, we mean that there exists $k\geq 3$ such that $\WNU(k)$ is satisfied).
Moreover, $\pol(\rel A)$ satisfies $\WNU(k)$ for \emph{all} $k\geq 3$ if, and only if, $\rel A$ has bounded width~\cite{BartoKozikConsistency}.
For a finite PCSP template $(\rel A,\rel B)$, it is known that if $(\rel A,\rel B)$ has bounded width, then $\pol(\rel A,\rel B)$ satisfies $\WNU(k)$ for all $k\geq 3$, although in that case the two properties are not equivalent~\cite{AtseriasDalmau}.

In the following, let $\rel H_r$ be the structure on an $r$-element set with the ternary relation containing all non-constant triples.
Let $\rel K_r$ be the complete graph on $r$ elements.
Let $\rel D$ be the digraph on $\{x,y,z\}$ with edges $(x,y),(y,z),(z,x),(x,z)$.
\begin{theorem}\label{thm:olsak-siggers}
	Let $\rel B$ be a structure and $\group$ be an oligomorphic subgroup of $\aut(\rel B)$. The following hold:
	\begin{itemize}
		\item If $\rel B$ does not pp-construct $(\rel H_2,\rel H_r)$ for any $r\geq 2$, then $\pol(\rel B)$ satisfies the Ol\v{s}\'ak condition modulo $\overline{\group}$.
		\item If $\rel B$ does not pp-construct $(\rel K_3,\rel K_r)$ for any $r\geq 3$, then $\pol(\rel B)$ satisfies the 6-ary Siggers condition modulo $\overline{\group}$.
		\item If $\rel B$ does not pp-construct $(\rel D,\rel K_r)$ for any $r\geq 3$, then $\pol(\rel B)$ satisfies the 4-ary Siggers condition modulo $\overline{\group}$.
		\item If $\rel B$ has bounded width, then $\pol(\rel B)$ satisfies the condition $\WNU(k)$ modulo $\overline{\group}$ for all $k\geq 3$.
	\end{itemize}
\end{theorem}
\begin{proof}
	We prove the contrapositive of the first item.
    Suppose that $\pol(\rel B)$ does not satisfy the Ol\v{s}\'ak condition modulo $\overline{\group}$.
    By~\Cref{sca}, there exists a finite subset $A$ of $B$ such that $\pol(\rel A, \rel B)$ does not satisfy the Ol\v{s}\'ak condition modulo $\overline{\group}$.
    Then by~\Cref{proj-limit}, there exists $m\geq 1$ such that $\pol(\mpow{\rel A},\mpow{\rel B}/{\group})$ does not contain an Ol\v{s}\'ak operation.
    By~\cite[Theorem 6.2]{BBKO}, we get that there exists a minion homomorphism $\pol(\mpow{\rel A},\mpow{\rel B}/{\group})\to\pol(\rel H_2, \rel H_r)$ for some $r\geq 2$.
    By composing with the local minion homomorphism $\pol(\rel B)\to \pol(\mpow{\rel A},\mpow{\rel B}/{\group})$ and using~\Cref{prop:pp-constructions-from-homo}, $\rel B$ pp-constructs $(\rel H_2,\rel H_r)$.
    
    The proof of the second and third items is similar as the one above, using~\cite[Theorem 6.9]{BBKO} for the second item; no corresponding statement for the third item seems to exist in the literature, but it can be proved exactly in the same way and we take it as folklore.
    We now prove the last item. Suppose that there exists a $k\geq 3$ such that $\pol(\rel B)$ does not satisfy the condition $\WNU(k)$ modulo $\overline{\group}$.
    By~\Cref{sca} and~\Cref{prop:lift-poly-oligo}, there exists a finite $A\subseteq B$ and $m\geq 1$ such that $\pol(\mpow{\rel A},\mpow{\rel B}/{\group})$ does not satisfy $\WNU(k)$.
    By~\cite{AtseriasDalmau}, we get that $(\mpow{\rel A},\mpow{\rel B}/{\group})$ does not have bounded width.
    Since $\rel B$ pp-constructs this template, it does not have bounded width either.
\end{proof}

It follows that if $\pol(\rel B)$ does not satisfy the Ol\v{s}\'ak condition modulo the closure of an arbitrary oligomorphic subgroup of $\aut(\rel B)$, then $\CSP(\rel B)$ is NP-hard by~\cite{DinurRegevSmyth}.
In the context of reducts of finitely bounded homogeneous structures, this gives a potentially stronger hardness criterion than the hardness border conjectured by Bodirsky and Pinsker.
We conjecture however that if $\rel B$ is an $\omega$-categorical model-complete core that pp-constructs $(\rel H_2,\rel H_r)$ for some $r\geq 2$, then $\rel B$ pp-constructs every finite structure and therefore is already NP-hard under the previously known hardness results.
Similarly, if $\pol(\rel B)$ does not satisfy the condition $\WNU(k)$ modulo an oligomorphic subgroup of $\aut(\rel B)$ for some $k\geq 3$, then $\CSP(\rel B)$ is not solvable by local consistency methods. This gives us~\Cref{thm:olsak-wnu}.
Finally, assuming that $\PCSP(\rel K_3,\rel K_r)$ is NP-hard for all $r\geq 3$, we obtain~\Cref{no-6-siggers-hard}.
This raises the natural question whether the assumptions in the second item of~\Cref{thm:olsak-siggers} and the main theorem in~\cite{TopoIrrelevant} are in fact equivalent, and in particular whether every $\omega$-categorical non-bipartite graph with finite chromatic number must pp-construct $\rel K_3$.
We observe that the statement easily follows from~\cite{BBKO} if $\rel G$ contains a triangle and has chromatic number at most $4$.

The proof of~\Cref{thm:olsak-siggers} in fact yields operations satisfying $\WNU(k)$ modulo $\overline{\aut(\rel B)}$ under the weaker assumption that $\PCSP(\rel A,\rel B)$ has bounded width for every finite substructure $\rel A$ of $\rel B$.
In the case where $\rel B$ is a temporal structure, this gives that $\pol(\rel B)$ satisfies $\WNU(k)$  for all $k\geq 3$ modulo every oligomorphic subgroup of $\aut(\rel B)$, whenever $\CSP(\rel B)$ is definable in fixpoint logic, by the results in~\Cref{sect:algorithms} and~\cite{TemporalDescriptive}.
We note that in the case of $\group=\aut(\mathbb Q;<)$ the result is not new, see e.g.~\cite{TemporalDescriptive} where explicit operations are given using the finer algebraic descriptions of polymorphism clones of temporal structures from~\cite{BodirskyKara}.
However, the fact that $\WNU(k)$ is satisfied for all $k\geq 3$ modulo the closure of an arbitrary oligomorphic subgroup of $\aut(\mathbb Q;<)$ seems new. %

\section{Proof of~\Cref{bw-not-finitely-tractable}}\label{sect:hardness}

In this section, we study the properties of the PCSP templates of the form $(\mpow{\rel A},\mpow{\rel B}/{\aut(\mathbb Q;<)})$, where $(\rel A,\rel B)$ is a PCSP template such that $\rel A$ is finite and $\rel B$ is a temporal structure.
We prove that these templates are not necessarily finitely tractable (even if the starting structure $\rel B$ has a tractable CSP), and not necessarily solvable by the BLP+AIP algorithm.
This yields in particular~\Cref{bw-not-finitely-tractable}.

To do so, we use the fact that there are known necessary conditions about $\pol(\rel A,\rel B)$ for finite tractability of $\PCSP(\rel A,\rel B)$ (or solvability by BLP+AIP).
For finite tractability, it follows from the cyclic term theorem~\cite{BartoKozikAbsorption} that if $\PCSP(\rel A,\rel B)$ is finitely tractable, then $\pol(\rel A,\rel B)$ contains cyclic operations of all large enough prime arities. Such an operation $f$ of arity $p$ satisfies all the identities of the form
$f(x_1,\dots,x_p)\approx f(x_{\sigma(1)},\dots,x_{\sigma(p)})$
for a cyclic permutation $\sigma$ of $[p]$.
Concerning BLP+AIP, it is necessary and sufficient for $\PCSP(\rel A,\rel B)$ to be solvable by BLP+AIP that $\pol(\rel A,\rel B)$ contains so-called \emph{2-block symmetric} polymorphisms of all odd arities.
Such an operation $f$ has arity $2L+1\geq 5$ and satisfies for arbitrary permutations $\sigma,\tau$ of $[L+1]$ and $[L]$ the identities

\vspace{-1.7em}
\begin{align*} &f(x_1,\dots,x_{L+1},y_1,\dots,y_{L}) \\\approx &f(x_{\sigma(1)},\dots,x_{\sigma(L+1)},y_{\tau(1)},\dots,y_{\tau(L)}) \end{align*}

In all the cases below, to prove that such operations cannot exist in $\pol(\mpow{\rel A},\mpow{\rel B}/{\aut(\rel B)})$, we first use~\Cref{lem:lift-poly} to reduce the task to disproving the existence of operations in $\pol(\rel A,\rel B)$ satisfying the relevant identities modulo $\overline{\aut(\rel B)}$.
Moreover, using the fact that every identity under consideration is given by permutations, we use the following result to further reduce the task to disproving the existence of operations in $\pol(\rel A,\rel B)$ properly satisfying the identities.

\begin{restatable}{lemma}{pseudonopseudo}\label{lem:pseudo-no-pseudo}
    Let $\rel A$ be a structure and let $\rel B$ be a temporal structure.
    Suppose that $f(x_1,\dots,x_n)\approx f(x_{\sigma(1)},\dots,x_{\sigma(n)})$ is satisfiable in $\pol(\rel A,\rel B)$
    modulo $\overline{\aut(\rel Q;<)}$, where $\sigma$ is a permutation of $\{1,\dots,n\}$.
    Then it is satisfiable in $\pol(\rel A,\rel B)$.
\end{restatable}

\begin{corollary}\label{cor:summary-necessary-conditions}
    Let $(\rel A,\rel B)$ be a PCSP template where $\rel A$ is a finite structure and $\rel B$ is a temporal structure.
    Let $m\geq 3$ be larger than the arity of the relations of $\rel A$.
    The following hold:
    \begin{itemize}
        \item If\/ $\PCSP(\mpow{\rel A},\mpow{\rel B}/{\aut(\rel B)})$ is finitely tractable, then $\pol(\rel A,\rel B)$ contains cyclic operations of every large enough prime arity;
        \item If\/ $\PCSP(\mpow{\rel A},\mpow{\rel B}/{\aut(\rel B)})$ is solvable by BLP+AIP, then $\pol(\rel A,\rel B)$ contains 2-block symmetric operations of all odd arities.
    \end{itemize}
\end{corollary}
\begin{proof}
	We use here the results from the previous section, in particular~\Cref{lem:lift-poly}.
	Recall that $(\mathbb Q;<)$ is a finitely bounded homogeneous structure with bounds of size $3$.

    Suppose first that $\PCSP(\mpow{\rel A},\mpow{\rel B}/{\aut(\rel B)})$ is finitely tractable.
    By~\cite{BartoKozikAbsorption}, this implies that $(\mpow{\rel A},\mpow{\rel B}/{\aut(\rel B)})$ has a cyclic polymorphism of every sufficiently large prime arity.
    By~\Cref{lem:lift-poly}, we obtain that $\pol(\rel A,\rel B)$ contains polymorphisms of sufficiently large prime arity satisfying the cyclic condition modulo $\overline{\aut(\rel B)}$.
    By~\Cref{lem:pseudo-no-pseudo}, such polymorphisms are in fact cyclic.

    Suppose now that $\PCSP(\mpow{\rel A},\mpow{\rel B}/{\aut(\rel B)})$ is solvable by BLP+AIP. Then by~\cite{BLPAIP}, for every $L\geq 1$ there exists a 2-block symmetric operation of arity $2L+1$ in $\pol(\mpow{\rel A},\mpow{\rel B}/{\aut(\rel B)})$.
    By~\Cref{lem:lift-poly}, the minor identities $\Sigma$ defining such operations are satisfied in $\pol(\rel A,\rel B)$ modulo $\overline{\aut(\rel B)}$.
    Finally, since the minor identities in $\Sigma$ are defined by permuting variables, they must be satisfied in $\pol(\rel A,\rel B)$ by~\Cref{lem:pseudo-no-pseudo}.
\end{proof}

Let $I\subseteq\mathbb Q^3$ contain the tuples $(a,b,c)\in\mathbb Q^3$ such that if $a=b$, then $a\geq c$
and $X\subseteq\mathbb Q^3$ contain  $(a,b,c)$ iff two entries are equal and the other one is greater than the others.
\begin{restatable}{proposition}{Inonothing}\label{prop:I-no-nothing}
    Let $\rel A$ be the substructure of $\rel B=(\mathbb Q;I,\neq)$ induced by $\{0,1\}$.
    Then $\pol(\rel A,\rel B)$ does not contain any cyclic operation or any 2-block symmetric operation of arity $\geq 5$.
\end{restatable}

\begin{restatable}{proposition}{sandwichX}\label{prop:no-sandwich-X}\label{prop:X-BLPAIP}
    Let $\rel A$ be the substructure of\/ $\rel B=(\mathbb Q;X)$ induced by $\{0,1,2\}$.
    Then $\pol(\rel A,\rel B)$ does not contain any cyclic operation of arity $1$ modulo $4$, nor 2-block symmetric polymorphisms of arity $7$.
\end{restatable}

With this in hand, we can prove~\Cref{bw-not-finitely-tractable}.
\omegasandwich*
\begin{proof}
	Let $\rel B$ be $(\mathbb Q; I, \neq)$, and $\rel A$ be the substructure of $\rel B$ induced by $\{0,1\}$.
	By~\Cref{prop:I-no-nothing}, $\pol(\rel A,\rel B)$ does not contain cyclic polymorphisms of any arity $\geq 2$, or any $2$-block symmetric polymorphism.
	Thus, $\PCSP(\mpow[3]{\rel A},\mpow[3]{\rel B}/{\aut(\rel Q;<)})$ is not finitely tractable nor solvable by BLP+AIP by~\Cref{cor:summary-necessary-conditions}.
	However, $\PCSP(\rel A,\rel B)$ is solvable by singleton arc consistency by~\Cref{pp-semilattice} and therefore it has width $4$.
	By the second item of~\Cref{relating-hierarchies} (applied with the minion only containing projections on a 2-element set), $\PCSP(\mpow[3]{\rel A},\mpow[3]{\rel B}/{\aut(\mathbb Q;<)})$  has width $4$.
	Finally, the template $(\mpow[3]{\rel A},\mpow[3]{\rel B}/{\aut(\mathbb Q;<)})$ admits as a sandwich the $\omega$-categorical structure $\mpow[3]{\rel B}$, whose CSP is in P, but is known not to have bounded width~\cite{FastDatalog}.
\end{proof}

We note that $\rel A$ in the proof above can be taken to be any finite substructure of $\rel B$.
If we let $\rel A_1\subseteq\rel A_2\subseteq\dots$ be finite structures such that $\bigcup\rel A_i=\rel B$, we obtain a sequence of finite PCSPs all with width $4$ and such that the limit of these problems, which is $\CSP(\mpow[3]{\rel B})$, does not have bounded width.

By taking $\rel B=(\mathbb Q;X)$, one arrives similarly at a finite PCSP template solvable by singleton AIP, not finitely tractable nor solvable by BLP+AIP, and with a tractable $\omega$-categorical sandwich.
We note that this template is not solvable by constantly many rounds of the Sherali-Adams hierarchy (the hierarchy of lift-and-project relaxations applied to the basic linear programming relaxation) either.
This follows from~\cite{TemporalDescriptive} and the fact that  constantly many rounds of the Sherali-Adams hierarchy can be implemented in fixpoint logic with counting~\cite{DawarHolm}.

\begin{proposition}\label{prop:sa-not-x}
    Let $\rel A$ be the substructure of $\rel B=(\mathbb Q;X)$ induced by $\{0,1\}$.
    Then $\PCSP(\rel A,\rel B)$ is not solvable by constantly many rounds of the Sherali-Adams hierarchy.
\end{proposition}
\begin{proof}[Proof sketch]
    	\cite{TemporalDescriptive} proves that $\CSP(\rel B)$ is not solvable in fixpoint logic with counting (FPC) by exhibiting, for every natural number $k\geq 1$, a pair of structures $\rel X, \rel Y$ that are equivalent in the counting logic $C^k$ and such that $\rel X\to\rel B$ and $\rel Y\not\to\rel B$.
        By inspection of the proof, it is the case that $\rel X\to\rel A$, and therefore $\PCSP(\rel A,\rel B)$ is not solvable in the logic FPC.
        Finally, we observe that the Sherali-Adams relaxation of an instance $\rel X$ is interpretable in first-order logic over $\rel X$.
        By Theorem 4.3 from~\cite{DawarHolm}, if $\PCSP(\rel A,\rel B)$ were solvable by the Sherali-Adams hierarchy, then it would be solvable in FPC.
        Thus, $\PCSP(\rel A,\rel B)$ is not solvable by constantly many rounds of the Sherali-Adams hierarchy.
\end{proof}

\bibliographystyle{alpha}
\bibliography{refs.bib}

\clearpage
\appendix

\section{Appendix}

\subsection{Algorithms for temporal CSPs}

\majsemilattice*
\begin{proof}
	Let $f\in\pol(\rel B,0,\Theta)$ be an operation whose image in $\mathscr C$ is the majority operation.
	Let $\alpha$ be an arbitrary automorphism of $(\mathbb Q;<)$ mapping $0$ to a positive value, and let $g(x,y):=f(x,y,\alpha(y))$.
	Note first that $g\in\pol(\rel B)$.
	Moreover, we have $g(0,0)=f(0,0,\alpha(0))=0$, and for all $x>0$ we get $g(0,x)=f(0,x,\alpha(x))>0$ and $g(x,0)=f(x,0,\alpha(0))>0$, and for all $x,y>0$ we have $g(x,y)>0$.
	Thus, $g$ is an operation that preserves $0$, it preserves $\Theta$ as both $f$ and $\alpha$ do,
	and its image in $\mathscr C$ is the desired semilattice.
\end{proof}

\begin{restatable}{lemma}{easylemma}\label{easylemma}
$\mathscr C$ is a subset of\/ $\pol(\rel B/{\Theta})$.
\end{restatable}
\begin{proof}
	Let $f\in\pol(\rel B,0,\Theta)$ of arity $k$, let $R$ be a relation symbol of arity $r$, and let $a^1,\dots,a^k\in R^{\rel A/{\Theta}}$.
	By definition of $\rel A/{\Theta}$, there exists $b^1,\dots,b^k\in R^{\rel A}$ such that $b^i_j=0\Leftrightarrow a^i_j=\{0\}$, for all $i\in\{1,\dots,k\}$ and $j\in\{1,\dots,r\}$.
	Since $f$ is a polymorphism of $\rel B$, we have $f(b^1,\dots,b^k)\in R^{\rel B}$.
	By assumption on $n$, we have $r\leq n$ and therefore there exists an automorphism $\alpha$ of $(\rel B,0)$ such that $\alpha f(b^1,\dots,b^k)\in A^r$, and since $\rel A$ is an induced substructure of $\rel B$ we obtain $\alpha f(b^1,\dots,b^k) \in R^{\rel A}$.
	For $j\in\{1,\dots,r\}$, let $c_j = \{0\}$ if the $j$th component of $\alpha f(b^1,\dots,b^k)$ is $0$, and $c_j=\mathbb Q_{>0}$ otherwise.
	Thus, the result of the function induced by $f$ on $\rel A/{\Theta}$ when applied to $a^1,\dots,a^k$ is $c$, and we have $c\in R^{\rel A/{\Theta}}$.
\end{proof}

\projectionstrategy*
\begin{proof}
	We prove that every $g\in G$ is a partial homomorphism.
	Suppose that $(x_{i_1},\dots,x_{i_r})\in\tilde R^{\rel X\setminus F}$, where $\tilde R$ is the symbol corresponding to the relation defined by $\exists_{y\in F} y\, (R(x_1,\dots,x_m))$ in $\rel B$, and where $x_{i_1},\dots,x_{i_r}\in\dom(g)$. 
	Thus, $(x_1,\dots,x_m)\in R^{\rel X}$.
	By definition of $G$ and since $k$ is large enough, there exists a partial homomorphism $h\in H$ with domain $\{x_1,\dots,x_m\}$ whose restriction to $X\setminus F$ is $g$.
	Thus, $(h(x_1),\dots,h(x_n))\in R^{\rel B}$, which implies that  $(h(x_{i_1}),\dots,h(x_{i_r}))\in \tilde R^{\rel B}$ and therefore $g$ is a partial homomorphism.
	The fact that $G$ is closed under restrictions and expansions as in the definition of $k$-strategies is clear from the fact that $H$ is, and the boundedness of $G$ is immediate.
	
	We now turn to the second item.
	Let $g\in G$ have domain containing $\{x_1,\dots,x_m\}$ and suppose that $(x_1,\dots,x_m)\in\tilde R^{\rel X/{F}}$ where $\tilde R$ is the symbol corresponding to the relation defined by $R(x_1,\dots,x_m)\land\bigwedge_{z,z'\in F}z=z'$ in $\rel B$.
	By definition of the contraction, we have $(x_1,\dots,x_m)\in R^{\rel X}$.
	Thus, $(g(x_1),\dots,g(x_m))\in R^{\rel B}$ and by assumption on $H$, we have $g(x)=g(y)$ whenever $x,y\in F$.
	It follows that $(g(x_1),\dots,g(x_m))\in \tilde R^{\rel B}$,
	and therefore $g$ is a partial homomorphism from $\rel X/{F}$ to $\rel B$.
	
	Concerning the third item, suppose that $\rel X$ is singleton arc consistent with respect to $\rel B$, and let $(D_x)_{x\in X}$ be a non-trivial potato system from $\rel X$ to $\rel B$ that witnesses this.
	By the same argument as in the first paragraph, $(D_x)_{x\in X\setminus F}$ is a potato system for the projection of $\rel X$ onto $X\setminus F$.
	To prove singleton arc consistency, let $x\in X\setminus F$ and $a\in D_x$.
	By assumption, there exists a potato system $(D'_y)_{y\in X}$ from $\rel X$ to $\mathbb B$ such that $D'_{y}\subseteq D_y$ for all $y\in X$ and $D'_x=\{a\}$.
	Again, $(D'_y)_{y\in X\setminus F}$ is a potato system for $\rel X\setminus F$, so we are done.
	Clearly, if the original potato system is bounded, so is its restriction to $X\setminus F$.
	
	Finally, any solution to $\AIP(\rel X,\rel A)$ can be restricted to a solution to $\AIP(\rel X\setminus F,\rel A)$ in the same way that a potato system or strategy for $\rel X$ can be restricted to a potato system or a strategy for $\rel X\setminus F$.
\end{proof}

\subsection{Proof of \Cref{prop:pp-constructions-from-homo}}

We give here a proof of~\Cref{prop:pp-constructions-from-homo}. To obtain a pp-construction as in the statement, we note that it suffices to have a map whose domain is the set of operations of arity at most $N$, where $N$ is a bound on the size of the domain and the relations of $\rel C$, and where the map is only required to satisfy $\xi(f^\sigma)=\xi(f)^\sigma$ when both $f^\sigma,f$ have arity at most $N$. For the purpose of the coming proof, we call such a map a \emph{partial minion homomorphism}.
Note that if $\rel A$ and $\rel C$ are finite then such a map is necessarily local.

\constructionshomo*
\begin{proof}
    Let $N$ be larger than the size of $S$ and the size of any relation in $\rel S$.
    Since $\xi$ is local, there exists a finite subset $T'$ of $A$ such that for all $f,g\in\pol(\rel A,\rel B)$ of arity at most $N$, if $f|_{T'}=g|_{T'}$ then $\xi(f)|_S = \xi(g)|_S$.

    Moreover, there exists a finite subset $T$ of $A$ containing $T'$ such that for all $f\in\pol(\rel T,\rel B)$ of arity at most $N$, there exists $f'\in\pol(\rel A,\rel B)$ such that $f|_{T}=f'|_{T}$.
    Indeed, fix an enumeration $a_1,a_2,\dots$ of $A\setminus T$, and consider $T_i=T'\cup\{a_1,\dots,a_i\}$ for all $i\geq 1$.
    Fix an arity $n\leq N$. 
    The relation $R_i$ on $B$ consisting of the $n$-ary polymorphisms $(\rel T')^n\to \rel B$ (seen as tuples indexed by $(T')^n$) that extend to a polymorphism $(\rel T_i)^n\to\rel B$ is invariant under $\aut(\rel B)$ and has arity $|T_i|^n$, so by $\omega$-categoricity there exists $i(n)\geq 1$ such that for all $j\geq i(n)$, we have $R_{i(n)}=R_j$.
    Let $i=\max_{n\leq N} i(n)$ and let $T=T'_i$.
    Then any $f\in\pol(\rel T,\rel B)$ of arity at most $N$ is in $\bigcap_{j\geq 1} R_j$ and therefore extends to an operation $f'\in\pol(\rel A,\rel B)$.

    Thus, we have a partial minion homomorphism $\zeta\colon\pol(\rel T,\rel B)^{(\leq N)}\to\pol(\rel S,\rel D)$ from the set of polymorphisms of $(\rel T,\rel B)$ of arity at most $N$,
    mapping an $f\in\pol(\rel T,\rel B)$ of arity at most $N$ to $\xi(f')|_S$ for an arbitrary extension $f'\in\pol(\rel A,\rel B)$ of $f$.
    By the choice of $T'$, $\zeta(f)$ does not depend on the choice of the extension.
    
    We therefore obtain a pp-construction of $(\rel S,\rel D)$ in $(\rel T,\rel B)$, by~\cite[Theorem 4.12]{BBKO}.
    Since $\rel T\to\rel A\to\rel B$, this gives a pp-construction of $(\rel S,\rel D)$ in $(\rel A,\rel B)$.
\end{proof}

\subsection{Categorical properties of $\mpow{}$}

For an arbitrary set $B$ and $i_1,\dots,i_k\in\{1,\dots,m\}$, recall that we write $\pr{i_1,\dots,i_k}$ for the map $B^m\to B^k$ defined by $(a_1,\dots,a_m)\mapsto (a_{i_1},\dots,a_{i_k})$.
If $\group$ is a permutation group on $B$, we use the same notation for the natural map $B^m/{\group}\to B^k/{\group}$.

\directedsystem*
\begin{proof}
	It suffices to prove the claim for every $m$ and $n=m+1$.
	The first item follows from the fact that $(\mpow{\rel A},\mpow{\rel B}/{\group})$ has a pp-construction in $(\rel A,\rel B)$.
	Explicitly, for $f\in\pol(\rel B)$ of arity $n$, one defines $\xi^*_m(f)$ to be the map
	\[ (a^1,\dots,a^n)\mapsto (f(a^1_1,\dots,a^n_1),\dots,f(a^1_m,\dots,a^n_m))/{\group}\]
	which can be seen to be a polymorphism of $(\mpow{\rel A},\mpow{\rel B}/{\group})$.
	Moreover, $\xi_m^*$ is a local minion homomorphism.
	
	We exhibit a local minion homomorphism $\xi_{m+1}\colon \pol(\mpow[m+1]{\rel A},\mpow[m+1]{\rel B}/{\group})\to\pol(\mpow{\rel A},\mpow{\rel B}/{\group})$.
	Given $f\in\pol(\mpow[m+1]{\rel A},\mpow[m+1]{\rel B}/{\group})$ of arity $n$, let $\xi^{m+1}_m(f)$ be defined by 
	\[(a^1,\dots,a^n) \mapsto \pr{1,\dots,m}f(b^1,\dots,b^n)/{\group},\]
	 where for all $i$,  $b^i$ is the $(m+1)$-tuple obtained by repeating the last entry of $a^i$.
\end{proof}

\functoriality*
\begin{proof}
We check the adjunction.
Let $h\colon \rel X\to\mpow{\rel B}$ be a homomorphism.
Define $\eta(h)=h'\colon\mlow{\rel X}\to\rel B$ by $h'([a,i]):=\pr{i}(h(a))$.
This is well defined, because $(a,p)\leftrightarrow (b,q)$ readily implies that $\pr{p}(h(a))=\pr{q}(h(b))$.
We first check that this is indeed a homomorphism.

Suppose that $([a_1,i(1)],\dots,[a_\ell,i(\ell)])\in R^{\mlow{\rel X}}$.
Then by definition there exists $b\in R_j^{\rel X}$ where $j\colon[\ell]\to[m]$ is such that $(a_\star,i(p))\sim^{\rel X} (b,j(p))$ for all $p\in [\ell]$.
We see $h$ as a homomorphism $\rel X\times[m]\to\mpow{\rel B}\times[m]$ acting on the first component and as the identity on the second component.
Note that the relation $\sim$ from the definition of $\mlow{\rel X}$ is a union of existentially definable relations, so we obtain that $(h(a_p),i(p))\sim^{\mpow{\rel B}} (h(b),j(p))$ for all $p\in[\ell]$.
It follows that $\pr{i(p)}(h(a_p)) = \pr{j(p)}(h(b))$, and since $h(b)\in R^{\mpow{\rel B}}_j$, we obtain $(h'([a_1,i(1)]),\dots,h'([a_\ell,i(\ell)]))\in R^{\rel B}$.

Conversely, given $g\colon\mlow{\rel X}\to\rel B$, define $g'\colon \rel X\to\mpow{\rel B}$ by $h(a)=(g([a,1]),\dots,g([a,m]))$.
We check that this is a homomorphism.
Suppose that $a\in R_i^{\rel X}$ for some $i\colon[\ell]\to[m]$.
Then we obtain $((a,i(1)),\dots,(a,i(\ell)))\in R^{\rel Y}$, so $([a,i(1)],\dots,[a,i(\ell)])\in R^{\mlow{\rel X}}$, and since $g$ is a homomorphism we get $(g[a,i(1)],\dots,g[a,i(\ell)])\in R^{\rel B}$ and therefore $h(a)\in R_i^{\mpow{\rel B}}$.
If now $(a,b)\in eq_{i,j}^{\rel X}$, the construction of $\mlow{\rel X}$ gives that $[a,i(p)]=[b,j(p)]$ for all $p\in\{1,\dots,\ell\}$.
It follows that $g[a,i(p)]=g[b,j(p)]$ holds for all $p\in\{1,\dots,\ell\}$ and thus $\pr{i(1),\dots,i(\ell)}h(a)=\pr{j(1),\dots,j(\ell)}h(b)$.
By the definition of $\mpow{}$, we get $(h(a),h(b))\in eq_{i,j}^{\mpow{\rel B}}$.

The maps $\eta\colon h\mapsto h'$ and $\mu\colon g\mapsto g'$ are inverses of each other, proving the bijection $\struct(\rel X,\mpow{\rel B}) \simeq \struct(\mlow{\rel X},\rel B)$.

It remains to prove that this bijection is natural in $\rel X$ and $\rel B$.
Suppose that we have homomorphisms $\alpha\colon\rel X\to \rel Y$, $\beta\colon\rel B\to\rel D$, and let $h\colon \rel Y\to\mpow{\rel B}$. We have natural functions $\struct(\rel Y,\mpow{\rel B})\to\struct(\rel X,\mpow{\rel D})$ and $\struct[](\mlow{\rel Y},\rel B)\to\struct[](\mlow{\rel X},\rel D)$.
Then $(\eta\circ \alpha)(h) = \eta(h\circ \alpha)$ is the function that maps $[a,i]\in\mlow{\rel X}$ to $\pr{i}(h(\alpha(a)))$.
Similarly, $(\alpha\circ \eta)(h)$ is the function that maps $[a,i]$ to $\eta(h)([\alpha(a),i])=\pr{i}(h(\alpha(a)))$, so we get $\eta\circ\alpha=\alpha\circ\eta$.
\end{proof}

\projlimit*
\begin{proof}
	Let $f\in\minion{}$.
	Define the sequence $f_m:=\zeta_m(f)\in\pol(\mpow{\rel A},\mpow{\rel B}/{\group})$.
	For $n\geq m$, we have $\xi^n_m(f_n) = \xi^n_m(\zeta_n(f)) = \zeta_m(f)=f_m$, so~\Cref{prop:lift-poly-oligo} applies and gives us a $f^*\in\pol(\rel A,\rel B)$ such that $\xi^\infty_m(f^*)=f_m$ for all $m\geq 1$.
	Define $\zeta^*(f)=f^*$.
	Thus, $\xi^\infty_m(\zeta^*(f)) = \xi^\infty_m(f^*)=f_m=\zeta_m(f)$ and therefore $\xi^\infty_m\circ\zeta^*=\zeta_m$.
	
	We now prove that $\zeta^*$ is a $(\trivgrp,\overline{\group})$-minion homomorphism.
	Let $\sigma\colon [p]\to [q]$ and suppose that $f^\sigma= g$ is true in $\minion{}$.
	Then $f_m^\sigma = g_m$ is true in $\pol(\mpow{\rel A},\mpow{\rel B}/{\group})$ for all $m\geq 1$.
	
	Suppose for contradiction that $u\circ (f^*)^\sigma\neq v\circ g^*$ for all $u,v\in\overline{\group}$.
	By definition of $\overline{\group}$ and~\Cref{sca}, one finds finite tuples $a^1,\dots,a^q$ of some length $m$ such that $(f^*)^\sigma(a^1,\dots,a^q)$ and $g^*(a^1,\dots,a_q)$ are not in the same orbit under $\group$.
	Thus, $\xi^\infty_m((f^*)^\sigma)(a^1,\dots,a^q)\neq\xi^\infty_m(g^*)(a^1,\dots,a^q)$,
	and therefore $\zeta_m(f)^\sigma \neq \zeta_m(g)$, contradicting our assumptions.
\end{proof}

\begin{restatable}{lemma}{polysmallerorbits}\label{lem:poly-smaller-orbits}
    Let $1\leq k\leq m$, let $\rel X,\rel B$ be structures.
    Let $h\colon\mpow{\rel X}\to\mpow{\rel B}/{\group}$ be a homomorphism and $i_1,\dots,i_k\in\{1,\dots,m\}$.
    There exists a homomorphism $h_k\colon\mpow[k]{\rel X}\to\mpow[k]{\rel B}/{\group}$ such that $\pr{i_1,\dots,i_k}\circ h=h_k\circ \pr{i_1,\dots,i_k}$.
\end{restatable}
\begin{proof}
For ease of notation, we do the proof here for the case where $i_r=r$ for all $r\in\{1,\dots,k\}$.
The general case follows from this simply by permuting the arguments of $h$.
    Define $h_k(a_1,\dots,a_k)$ by $\pr{1,\dots,k}(h(a_1,\dots,a_k,a_{k+1},\dots,a_m))$ for an arbitrary choice of elements $a_{k+1},\dots,a_m\in X$.
    We show that the definition does not depend on the choice of these elements.
    For arbitrary $a'_{k+1},\dots,a'_m$, the pair
    \[((a_1,\dots,a_k,a_{k+1},\dots,a_m),(a_1,\dots,a_k,a'_{k+1},\dots,a'_m))\]
    is in the interpretation of the symbol $eq_{i,j}$ in $\mpow{\rel X}$ with $i,j\colon\{1,\dots,k\}\to\{1,\dots,m\}$ being both equal to the inclusion map.
    So $(h(a_1,\dots,a_m),h(a_1,\dots,a_k,a'_{k+1},\dots,a'_m))$ is in the interpretation of the symbol $eq_{i,j}$ in $\mpow{\rel B}/{\group}$, so that there are tuples $b\in h(a_1,\dots,a_m),c\in h(a_1,\dots,a_k,a'_{k+1},\dots,a'_m)$ such that $(b,c)\in eq^{\mpow{\rel B}}_{i,j}$, i.e., $(b_1,\dots,b_k)=(c_1,\dots,c_k)$.
    Thus, \[\pr{1,\dots,k}(h(a_1,\dots,a_m))=\pr{1,\dots,k}(h(a_1,\dots,a_k,a'_{k+1},\dots,a'_m))\]
    as desired.
\end{proof}

\gradedadjoint*
\begin{proof}
	Let $(H,\xi)$ be an arrow in $\MSystems{m\cdot k}(\mlow{\rel X},\rel B)$.
	For every subset $K\subseteq X$ of size at most $k$, $\mlow{K}$ has size at most $mk$ and we get that $H_{\mlow{K}}$ is a set of homomorphisms $\mlow{\rel K}\to\rel B$, where we write $\rel K$ for the substructure of $\rel X$ induced by $K$.
	Let $G_K$ be the image of $H_{\mlow{K}}$ under the bijection $\nu\colon\struct[](\mlow{\rel K},\rel B)\to\struct[](\rel K,\mpow{\rel B})$,
	and let $\zeta_K = (\xi_{\mlow{K}})^\nu$.
	Then $(G,\zeta)$ is an arrow in $\MSystems{k}(\rel X,\mpow{\rel B})$.
	
	Conversely, suppose $(G,\zeta)$ is an arrow in $\MSystems{k}(\rel X,\mpow{\rel B})$.
	Let $K\subseteq\mlow{\rel X}$ have size at most $k$.
	Let $L\subseteq \rel X$ be the smallest set such that $K\subseteq\mlow{L}$; we have that $L$ also has size at most $k$.
	Let $H_K$ be the image of $G_L$ under the function $\mu$ obtained by composing the bijection $\struct[](\rel L,\mpow{\rel B})\to\struct[](\mlow{\rel L},\rel B)$ with the restriction $\struct[](\mlow{\rel L},\rel B)\to\struct[](\rel K,\rel B)$.
	Finally, define $\xi_{K}:=(\zeta_{L})^\mu$.
\end{proof}

\category*
\begin{proof}
	Let $a\in\minion\{1\}$ be arbitrary.
	The identity in $\MSystems{k}(\rel A,\rel A)$ is given by the system where $H_K=\{\mathrm{id}_K\}$ for $K\in\binom{A}{\leq k}$
	and defining $\xi_K = (\minion f)(a)$,  with $f$ being the bijection $\{1\}\to H_K$.
	
	Suppose that $(F,\xi)\colon \rel A\to\rel B$ and $(G,\zeta)\colon\rel B\to\rel C$ are arrows in $\MSystems{k}$.
	Define $H=\{g\circ f\mid f\in F, g\in E, \im(g)\subseteq\dom(g)\}$, and for $K\in\binom{C}{\leq k}$ let $\chi_K$ be the $H_K\times[d]$ matrix whose row $h\in H_K$ is defined to be
	$\chi_K(h)=\sum_{f\in F_K} \sum_{g\in G_{f(K)}, h=g\circ f} \xi_K(f)\odot \zeta_{f(K)}(g)$.
	By assumption on \minion{}, this is an element of $\minion H_K$.
	
	We show that $(\chi_L)^\sigma=\chi_K$ whenever $K\subseteq L\in\binom{A}{k}$ and $\sigma$ is the canonical restriction $H_L\to H_K$.
	For $h\in H_K$, we have
	\begin{align*}
	(\chi_L)^\sigma(h) &= \sum_{\substack{h'\in H_L\\h'|_K=h}} \chi_L(h')\\
			&=  \sum_{\substack{h'\in H_L\\h'|_K=h}} \sum_{f'\in F_L} \sum_{\substack{g'\in G_{f'(L)}\\h'=g'\circ f'}} \xi_L(f') \odot \zeta_{f'(L)}(g')\\
			&= \sum_{\substack{f\in F_K\\g\in G_{f(K)}\\h=g\circ f}}\sum_{\substack{f'\in F_L\\f'|_K=f}} \sum_{\substack{g'\in G_{f(L)}\\g'|_{f(K)}=g}} \xi_L(f') \odot \zeta_{f'(L)}(g'),
	\end{align*}
	where to go from the second to third equation we simply grouped the functions $f'$ (resp.\ $g'$) according to their restrictions on $K$ (resp.\ $f(K)$).
	Continuing, this gives
	\begin{align*}
			&= \sum_{\substack{f\in F_K\\g\in G_{f(K)}\\h=g\circ f}}\left(\sum_{\substack{f'\in F_L\\f'|_K=f}}  \xi_L(f')\right)\odot\left(\sum_{\substack{g'\in G_{f(L)}\\g'|_{f(K)}=g}} \zeta_{f(L)}(g')\right)\\
			&= \sum_{\substack{f\in F_K\\g\in G_{f(K)}\\h=g\circ f}} (\xi_L)^{\sigma}(f)\odot(\zeta_{f(L)})^\sigma(g)\\
			&= \sum_{\substack{f\in F_K\\g\in G_{f(K)}\\h=g\circ f}}\xi_K(f)\odot\zeta_{f(K)}(g) = \chi_K(h).
	\end{align*}
	Thus, $(H,\chi)$ is an arrow in $\MSystems{k}(\rel A,\rel C)$.
\end{proof}

\reductionfinite*
\begin{proof}
    If $h\colon\rel X\to\rel A$, then we have $\mpow{\rel X}\to\mpow{\rel A}$ by~\Cref{functoriality}.

    Suppose now that there exists a homomorphism $h\colon\mpow{\rel X}\to\mpow{\rel B}/{\aut(\rel D)}$.
    We define a structure $\rel Y$ in the same signature as $\rel D$ as follows.
    
    By~\Cref{lem:poly-smaller-orbits}, $h$ induces a homomorphism $h_k\colon\mpow[k]{\rel X}\to\mpow[k]{\rel B}/{\aut(\rel D)}$ for all $1\leq k\leq m$, obtained by considering the restriction of an $m$-tuple to its first $k$ coordinates.
    Let $\sim$ be the equivalence relation on $I$ defined by $a\sim b$ iff $h_2(a,b)$ is an orbit of constant pairs.
    The domain of $\rel Y$ is $X/{\sim}$, and we denote the equivalence class of $a\in X$ by $[a]$.

    For every relation $R$ of $\rel D$ of arity $r$, let $([a_1],\dots,[a_r])$ be in $R^{\rel Y}$ iff the tuples in $h_r(a_1,\dots,a_r)$ satisfy the atomic formula $R(x_1,\dots,x_r)$ in $\rel D$.
    Using $m>r$, we show that this does not depend on the choice of representative for the classes $[a_1],\dots,[a_r]$.
    Suppose that the tuples in $h_r(a_1,\dots,a_r)$ satisfy $R(x_1,\dots,x_r)$, and that $a_1\sim a'_1$.
    Then the tuples in $h_{r+1}(a_1,\dots,a_r,a'_1)$ satisfy $R(x_1,\dots,x_r)$ and $x_1=x_{r+1}$,
    so they satisfy $R(x_{r+1},x_2,\dots,x_r)$.
    This shows that the choice of the representative for $[a_1]$ is irrelevant, and similarly for any other $[a_i]$.

    We check that the structure $\rel Y$ does not contain any obstruction.
    For every substructure $\rel K$ of $\rel Y$ induced by elements $[a_1],\dots,[a_m]$,
    let $(b_1,\dots,b_m)\in D^m$ be elements in the orbit $h(a_1,\dots,a_m)$.
    Then $e$ induces an embedding $e'\colon\rel K\hookrightarrow \rel D$ defined by $e'(c)=b_i$ if $e(c)=[a_i]$.
    Since the bounds of $\rel D$ have size at most $m$, this implies that there exists an embedding $g\colon \rel Y\hookrightarrow \rel D$, that we can see as a function $f\colon I\to C$ by precomposing with the canonical projection $I\to I/{\sim}$.
    It remains to check that $f$ is a homomorphism $\rel X\to\rel B$, and that $\pi(f)=h$.

    Suppose that $(a_1,\dots,a_r)\in R^{\rel X}$.
    Let $i(j)=j$ for all $j\in [r]$.
    Then $h_r(a_1,\dots,a_r)$ is in the relation $R_i$ in $\mpow[r]{\rel B}/{\aut(\rel D)}$.
    Since $\rel B$ is a reduct of $\rel D$, $R$ is definable by a quantifier-free formula $\psi$ over $\rel D$.
    Then $h_r(a_1,\dots,a_r)$ is an orbit of tuples satisfying $\psi$ in $\rel D$,
    and therefore $([a_1],\dots,[a_r])$ satisfies $\psi$ in $\rel Y$.
    As $g$ is an embedding, $(g([a_1]),\dots,g([a_r]))$ satisfies $\psi$ in $\rel D$, i.e., $(f(a_1),\dots,f(a_r))=(g([a_1]),\dots,g([a_r]))\in R^{\rel B}$.

    Finally, let $(a_1,\dots,a_n)\in\mpow{\rel X}$.
    Then
    \begin{align*}
        \pi(f)(a_1,\dots,a_m)&=(f(a_1),\dots,f(a_m))/{\aut(\rel D)}\\
        &= (g([a_1]),\dots,g([a_m]))/{\aut(\rel D)}
    \end{align*}
    and $g([a_1],\dots,[a_m])\in h(a_1,\dots,a_m)$ since $g$ is an embedding $\rel Y\hookrightarrow\rel D$.
    Thus, $\pi(f)(a_1,\dots,a_m)=h(a_1,\dots,a_m)$ and we are done.
\end{proof}

\begin{corollary}
	Let $(\rel A,\rel B),(\rel C,\rel D)$ be PCSP templates and let $\mathscr G\leq\aut(\rel B)$ and $\mathscr H\leq\aut(\rel D)$ be oligomorphic subgroups. The following are equivalent:
	\begin{itemize}
		\item There exists a local $(\trivgrp,\overline{\mathscr H})$-minion homomorphism $\pol(\rel A,\rel B)\to\pol(\rel C,\rel D)$;
		\item $(\rel A,\rel B)$ pp-constructs every $(\mpow{\rel S},\mpow{\rel D}/{\mathscr H})$ for finite $\rel S\subseteq\rel C$ and $m\geq 1$.
	\end{itemize}
\end{corollary}
\begin{proof}
Suppose that $\xi\colon\pol(\rel A,\rel B)\to\pol(\rel C,\rel D)$ is a local $(\trivgrp,\overline{\mathscr H})$-minion homomorphism.
By composition with the local $(\overline{\mathscr H},\trivgrp)$-minion homomorphisms $\pol(\rel C,\rel D)\to\pol(\mpow{\rel S},\mpow{\rel D}/{\mathscr H})$,
we get a local $(\trivgrp,\trivgrp)$-minion homomorphism $\pol(\rel A,\rel B)\to\pol(\mpow{\rel S},\mpow{\rel D}/{\mathscr H})$ and therefore the second item holds.

In the other direction, the assumption implies that there are minion homomorphisms $\pol(\rel A,\rel B)\to\pol(\mpow{\rel S},\mpow{\rel D}/{\mathscr H})$ for all $m\geq 1$ and finite $\rel S\subseteq\rel C$.
By~\Cref{sca}, we obtain minion homomorphisms $\pol(\rel A,\rel B)\to\pol(\mpow{\rel C},\mpow{\rel D}/{\mathscr H})$ for all $m\geq 1$.
Indeed, let $\rel S_1\subseteq\rel S_2\subseteq\dots$ be such that $\bigcup\rel S_i=\rel C$.
Let $\xi_i$ be a minion homomorphism $\pol(\rel A,\rel B)\to\pol(\mpow{\rel S_i},\mpow{\rel D}/{\mathscr H})$.
By~\Cref{sca}, for every $f\in\pol(\rel A,\rel B)$, there exists an $f^*\colon (C^m)^n\to D$ that locally agrees, up to composition by elements of $\mathscr H$, with $\xi_i(f)$ for infinitely many $i$.
The map $f\mapsto \xi^\infty_m(f^*)$ is then a well-defined minion homomorphism $\pol(\rel A,\rel B)\to\pol(\mpow{\rel C},\mpow{\rel D}/{\mathscr H})$.
By~\Cref{proj-limit}, there exists a $(\trivgrp,\overline{\mathscr H})$-minion homomorphism $\pol(\rel A,\rel B)\to\pol(\rel C,\rel D)$.
\end{proof}

\begin{proposition}
    Let $\rel B$ be a first-order reduct of an $\omega$-categorical structure $\rel B^+$ with the Ramsey property, and let $m\geq 1$.
    Suppose that $(\rel A,\rel C)$ is a finite homomorphic relaxation of $\mpow{\rel B}$.
    Then there exists a finite substructure $\rel A'$ of $\rel B$ such that $(\rel A,\rel C)$ is a homomorphic relaxation of $(\mpow{\rel A'}, \mpow{\rel B}/{\aut(\rel B^+)})$.
\end{proposition}
\begin{proof}
    By assumption, there exists a homomorphism $h\colon\rel A\to\mpow{\rel B}$.
    Let $A'$ be the finite set containing all elements appearing in a tuple in $h(A)$, and let $\rel A'$ be the substructure of $\rel B$ induced by $A'$.
    Then we have a homomorphism $\rel A\to \mpow{\rel A'}$.

    Now, let $g$ be a homomorphism $\mpow{\rel B}\to\rel C$.
    Since $\rel D$ has the Ramsey property, there exists a homomorphism $g'\colon\mpow{\rel B}\to\rel C$ that is canonical with respect to $\aut(\rel D)$.
    This induces a homomorphism $\tilde g\colon\mpow{\rel B}/{\aut(\rel D)}\to \rel C$, so that we obtain the sequence
    \[\rel A\to \mpow{\rel A'}\to \mpow{\rel B}\to \mpow{\rel B}/{\aut(\rel D)}\to\rel C\]
    and concludes the proof.
\end{proof}

\subsection{Identities}

We conjectured above that an $\omega$-categorical graph $\rel G$ containing a triangle and with finite chromatic number must pp-construct all finite structures, or equivalently the complete graph with 3 elements $\rel K_3$.
We prove here our remark that the result is true for graphs with chromatic number $\leq 4$.
\begin{proposition}\label{fourchromatic}
Let $\rel G$ be an $\omega$-categorical graph containing a triangle and with chromatic number at most $4$.
Then $\rel G$ pp-constructs $\rel K_3$.
\end{proposition}
\begin{proof}
	We have $\rel K_3\to\rel G\to\rel K_4$ by assumption, and therefore there exists a local minion homomorphism $\pol(\rel G)\to\pol(\rel K_3,\rel K_4)$.
	It is also known~\cite[Example 2.22]{BBKO} that there exists a minion homomorphism $\pol(\rel K_3,\rel K_4)\to\pol(\rel K_3)$, and therefore we obtain by composition a local minion homomorphism $\pol(\rel G)\to\pol(\rel K_3)$. By~\Cref{prop:pp-constructions-from-homo}, it follows that $\rel K_3$ is pp-constructible in $\rel G$.
\end{proof}
The proof strategy for~\Cref{fourchromatic} cannot be generalized to any higher chromatic number, since it is known that $\pol(\rel K_3,\rel K_5)$ satisfies identities that are not satisfied in~$\pol(\rel K_3)$~\cite[Example 3.4]{BBKO}.

\subsection{Hardness of PCSP templates associated with temporal CSPs}

\pseudonopseudo*
\begin{proof}
    Let $k$ be the order of $\sigma$, and let $a_1,\dots,a_n\in A$.
    Then for all $i\in\{0,\dots,k-1\}$, all the pairs
    $(f(a_{\sigma^i(1)},\dots,a_{\sigma^{i}(n)}),f(a_{\sigma^{i+1}(1)},\dots,a_{\sigma^{i+1}(n)}))$ are in the same orbit under $\aut(\rel Q;<)$.
    This orbit cannot be that of an increasing pair, for we would obtain
    \begin{align*}
     f(a_1,\dots,a_n) &< f(a_{\sigma(1)},\dots,a_{\sigma(n)})<\dots\\
      &<f(a_{\sigma^{k-1}(1)},\dots,a_{\sigma^{k-1}(n)})<f(a_1,\dots,a_n),
      \end{align*}
    a contradiction. Similarly, the orbit cannot be that of a decreasing pair.
    So this orbit has to be the orbit of pairs with equal elements, i.e., $f(a_1,\dots,a_n)=f(a_{\sigma(1)},\dots,a_{\sigma(n)})$ holds for all $a_1,\dots,a_n$.
\end{proof}

Let $\rel B$ be $(\mathbb Q;I)$ and $\rel A$ be the substructure induced by $\{0,1\}$.
Let $f$ be an $n$-ary polymorphism of $(\rel A,\rel B)$.
A coordinate $i\in\{1,\dots,n\}$ is called \emph{essential} if there exists a pair $a,a'$ of tuples such that $f(a) < f(a')$ and such that $a_j=a'_j$ for all $j\neq i$.
Since $f$ preserves $I$ and therefore $\leq$, it must be that $a_i=0$ and $a'_i = 1$.
    We call $a$ a \emph{witnessing tuple} for $i$.

\begin{lemma}\label{lem:I-ess-coordinate}
    Let $\rel A$ be the substructure of $\rel B=(\mathbb Q;I)$ induced by $\{0,1\}$.
    Let $f\in\pol(\rel A,\rel B)$, let $i$ be an essential coordinate of $f$, and let $a$ be a witnessing tuple for $i$.
    Then for all $b_1,\dots,b_n\in A$, we have $f(a)\neq f(b_1,\dots,b_{i-1},1,b_{i+1},\dots,b_n)$.
\end{lemma}
\begin{proof}
    Consider the matrix
    \[\begin{pmatrix}
    b_1 & \dots & b_{i-1} & 1 & b_{i+1} & \dots & b_n\\
    a_1 & \dots & a_{i-1} & 0 & a_{i+1} & \dots & a_n\\
    a_1 & \dots & a_{i-1} & 1 & a_{i+1} & \dots & a_n
    \end{pmatrix}\]
    whose columns all belong to the relation $I$.
    If $c=(c_1,c_2,c_3)$ is the tuple resulting from applying $f$ to all the rows, then $c\in I$ and we know that $c_2<c_3$ since $a$ is a witnessing tuple for $i$. Thus, $c_1\neq c_2$ and the claim is proved.
\end{proof}

\Inonothing*
\begin{proof}
We first consider the case of cyclic polymorphisms.
    Suppose that $\pol(\rel A,\rel B)$ has a cyclic polymorphism $f$ of arity $n\geq 2$.
    Note that by preservation of $\neq$, $f$ cannot be a constant operation and therefore it depends on all its arguments.

   Let $i\in\{1,\dots,n\}$.
    We now claim that the only witnessing tuple for $i$ must be $(0,\dots,0)$.
    Indeed, suppose that $a$ is a witnessing tuple for $i$ and contains a $1$ at position $j$ (and $j\neq i$ since $a_i=0$ as we observed above).
    Let $b_{k}$ be $a_{k+j-i}$ for all $k\in\{1,\dots,n\}$, where the computation of indices is modulo $n$.
    Then we get $f(a_1,\dots,a_n)\neq f(b_1,\dots,b_{i-1},1,b_{i+1},\dots,b_n)$ by~\Cref{lem:I-ess-coordinate}.
    However, $f(b_1,\dots,b_{i-1},1,b_{i+1},\dots,b_n)$ is by definition $f(a_{1+j-i},a_{2+j-i},\dots,a_{j-1},1,a_{j+1},\dots,a_{n+j-i})$.
    Since $f$ is cyclic, this is the same as $f(a_1,a_2,\dots,a_{j-1},1,a_{j+1},\dots,a_{n})$, which is $f(a_1,\dots,a_n)$.
    This is a contradiction.

    Finally, it follows that if $a$ is any non-zero tuple, and $b$ is obtained by replacing a single $0$ by a $1$ in $a$ (say at a position $i$), then $f(a)=f(b)$.
    Indeed, otherwise $a$ would be a witnessing tuple for $i$, a contradiction to the previous paragraph.
    By transitivity, if $a$ is any non-zero tuple and $b\geq a$ (componentwise), then $f(a)=f(b)$.

    In particular, we must have $f(1,\dots,1,0,\dots,0)=f(1,\dots,1)=f(0,\dots,0,1,\dots,1)$ where the block of $0$s in the left-hand side has size $\lfloor n/2\rfloor$ and the one on the right-hand side has size $\lceil n/2\rceil$.
    This contradicts the fact that $f$ preserves $\neq$.

    We now turn to 2-block symmetric polymorphisms. We first remark that it suffices to prove that no 5-ary 2-block symmetric operation is a polymorphism to rule out any 2-block symmetric polymorphism of arity $\geq 5$.
    Indeed, if $f$ is  $2$-block symmetric of arity $2L+1$ and $L\geq 3$, then $g$ defined by
    \[g(x_1,\dots,x_{L},y_1,\dots,y_{L-1})=f(y_1,\dots,y_{L-1},0,1,x_1,\dots,x_L)\]
    is $2$-block symmetric polymorphism of arity $2L-1$, since the tuples $(0,0,0),(1,1,1)$ are in $I$ and since the symmetry of $g$ in the first block implies that $g$ preserves $\neq$.

    Let $f\in\pol(\rel A,\rel B)$ be such a polymorphism of arity $5$.
    By preservation of $\neq$, this polymorphism cannot be a constant operation and therefore it depends on at least one argument.
    It must also depend on arguments in both blocks, otherwise $f$ would induce a cyclic operation of arity $2$ or $3$, in contradiction with the first part of the proof.

    Then by~\Cref{lem:I-ess-coordinate}, for any witnessing tuple $a$ for $1$, and every $b_2,\dots,b_5$, we have $f(a)\neq f(1,b_2,\dots,b_5)$.

    Suppose that $a$ contains a $1$ in the first block, and because of 2-block symmetry we can assume that it is in the second position.
    Define $b_2=a_3, b_3=a_1, b_4=a_4, b_5=a_5$. Then $f(a)\neq f(1,a_3,a_1,a_4,a_5)=f(a_1,1,a_3,a_4,a_5)=f(a)$ by 2-block symmetry, which is a contradiction.
    Therefore, any witnessing tuple $a$ for a coordinate in the first block consists of 0s in the first block. A symmetric argument gives the same for the second block.

    We obtain the following weaker version of the property from the first part of the proof: for any $a$ that is non-zero in \emph{both blocks}, and any $b$ that is componentwise lower-bounded by $a$, we have $f(a)=f(b)$.

    We thus get $f(1,0,0,1,0)=f(1,1,1,1,1)=f(0,1,1,0,1)$, in contradiction to the fact that $f$ preserves  the relation $\neq$.
\end{proof}

\sandwichX*
\begin{proof}
    We start with disproving the existence of cyclic operations.
    By contradiction, let $f$ be such an operation of arity $2n+1$ (with $n$ even).
    By~\Cref{lem:pseudo-no-pseudo}, a pseudo-cyclic operation modulo $\overline{\aut(\rel C)}$ must be cyclic.
    For $k\in\{0,\dots,2n+1\}$ and $b\in\{0,1,2\}$, let us write $[b]*k$ for any tuple with $k$ consecutive $b$ (where we consider the first entry of a row to be adjacent with the last entry), the rest of the entries being 0.
    Since $f$ is cyclic, it is constant on such tuples for any $b$ and $k$.

    Note that whenever $k+\ell+i=2n+1$, then one can find a $3\times(2n+1)$ matrix whose rows are tuples of the form $[b]*k$, $[b']*\ell$, $[b']*i$ and whose columns are all in $X$, regardless of the values of $b,b'\in\{1,2\}$.
    Therefore, $(f([b]*k),f([b]*\ell),f([b']*i))\in X$.
    In the particular case where $k=\ell$, then we obtain $f([b]*k)<f([b']*i)$.

    We claim that for all $0\leq k < n/2$ and all $b,b'>0$, we have
    \begin{align*}f([b]*(2k+1)) &= f([b']*(2k+3)) \\&< f([1]*(2n-4k-1))\\&<f([2]*(2n-4k-3)).
    \end{align*} For $k=0$, we obtain $f([2]*1)=f([1]*3)$ and for $k=\lfloor n/2 \rfloor -1$ we get $f([1]*3)<f([2]*1)$, giving the desired contradiction. The proof is by induction on $k$. Assume $k=0$.
    Then since $1+1+(2n-1)=2n+1$, we have
    
    \begin{equation} \label{eq:fin_tractable_aux1}
    f([b]*1)<f([b']*(2n-1)) \quad \text{ for all } b,b'>0.
    \end{equation}
    
    Then by considering the matrix
    \[ \begin{pmatrix}1 & \dots & 1 & 1 & 1 & 0 & 0\\
    1 & \dots & 1 & 0 & 0 & 1 & 1\\
    2 & \dots & 2 & 0 & 0 & 0 & 0\end{pmatrix}\]
    whose first two rows have $2n-1$ consecutive 1, whose last row has $2n-3$ consecutive 2, and whose columns all belong to $X$, one obtains $f([1]*(2n-1))<f([2]*(2n-3))$. This, together with (\ref{eq:fin_tractable_aux1}) yields that
    $f([b]*1)< f([2]*(2n-3))$ for all $b>0$.
    Since $(2n-3)+3+1=2n+1$, we have 
    $f([b]*1)=f([b']*3)$ for all $b,b'>0$, which concludes the base case of the induction. \par

    Suppose now that $k>0$.
    From $(2k+1)+(2k+1)+(2n-4k-1)=2n+1$, we get that $f([b]*(2k+1))<f([b']*(2n-4k-1))$ for all $b,b'>0$.
    If $k=2\ell$ is even, then the induction hypothesis applied from $\ell$ up to $k-1$ gives the equality $f([b]*(2\ell+1))=f([b]*(2k+1))$.
    Furthermore, $(2\ell+1)+(2\ell+1)+(2n-2k-1)=2n+1$ gives that $f([b]*(2k+1)) < f([b']*(2n-2k-1))$ for all $b,b'>0$.
    If $k=2\ell+1$ is odd, then the induction hypothesis from $\ell$ up to $k-1$ gives $f([b]*(2\ell+1)) = f([b]*(2\ell+3))=f([b]*(2k+1))$.
    Then $(2\ell+1)+(2\ell+3)+(2n-2k-1)=2n+1$ gives $f([b]*(2k+1))<f([b']*(2n-2k-1))$.

    Considering now the matrix
    \[ \begin{pmatrix}1 & \dots & 1 & 1 & \dots & 1 & 0 & \dots & 0\\
    1 & \dots & 1 & 0 & \dots & 0 & 1 & \dots & 1\\
    2 & \dots & 2 & 0 & \dots & 0 & 0 & \dots & 0\end{pmatrix}\]
    whose first two rows have $2n-2k-1$ consecutive $1$, and whose last row has $2n-4k-3$ consecutive $2$, we obtain $f([1]*(2n-2k-1)) < f([2]*(2n-4k-3))$, and combining with the inequality in the previous paragraph this gives $f([b]*(2k+1))<f([2]*(2n-4k-3))$.

    Finally, by $(2k+1)+(2k+3)+(2n-4k-3)=2n+1$, we obtain $f([b]*(2k+1))=f([b']*(2k+3))$.
    
    Suppose now for contradiction that there exists a 2-block symmetric polymorphism $g$ of $(\rel A,\rel B)$ of arity $7$.
  Note that $f(x_1,\dots,x_4):=g(x_1,x_2,x_3,x_4,1,0,0)$ is a symmetric polymorphism of $(\rel A,\rel B)$ of arity $4$: it is a symmetric function since $g$ is 2-block symmetric, and it is a polymorphism since the tuples $(1,0,0),(0,1,0),(0,0,1)$ are in the relation of the template $\rel A$.
  We prove that such an $f$ cannot exist.
  Consider the following matrix with $b_1,\dots,b_4\in\{1,2\}$:
	\[\begin{pmatrix}
      b_1 & b_2 & 0 & 0\\
      0 & 0 & b_3 & b_4\\
      0 & 0 & 0 & 0
      \end{pmatrix}\]
    Note that the columns are in $X$.
    For any $b,b'\in\{1,2\}$, 
    by taking $b_1=b_3=b$ and $b_2=b_4=b'$, this gives us that $f(b,b',0,0)<f(0,0,0,0)$ holds for all $b,b'\in\{1,2\}$.
    By preservation of $X$ also $f(b_1,b_2,0,0)=f(b_3,b_4,0,0)$ holds for all $b_1,\dots,b_4\in\{1,2\}$.
    Then we obtain a contradiction by considering the matrix
    \[\begin{pmatrix}
      1 & 1 & 0 & 0 \\
      1 & 0 & 1 & 0 \\
      2 & 0 & 0 & 1 
      \end{pmatrix}\]
    whose columns are in $X$, but where the result of applying $f$ rowwise gives a constant triple, which is not in $X$.
\end{proof} 
\end{document}